\documentclass[12pt,oneside,english]{amsart}
\usepackage{CJKutf8}
\usepackage[a4paper,margin=1in]{geometry}
\usepackage{footnote}
\usepackage{color}
\usepackage{babel}
\usepackage{array}
\usepackage{subfig}
\usepackage{url}
\usepackage{amssymb}
\usepackage{verbatim}
\usepackage{hyperref}
\usepackage[foot]{amsaddr}
\usepackage{enumitem}
\usepackage{multirow}
\usepackage{amsmath}
\usepackage{amsthm}
\usepackage{graphicx}
\usepackage{setspace}
\usepackage[authoryear]{natbib}
\usepackage[all]{xy}
\usepackage{booktabs}
\usepackage{multibib}
\usepackage{booktabs}
\numberwithin{table}{section}
\theoremstyle{plain}
\newtheorem{prop}{\protect\propositionname}

\newtheorem{thm}{Theorem}
\newtheorem{lem}{\protect\lemmaname}
\newtheorem{example}{Example}
  
  \theoremstyle{definition}
  \newtheorem{defn}{\protect\definitionname}
\makeatletter
\hypersetup{
    colorlinks=true,
    linkcolor=blue,
    citecolor=blue,      
    urlcolor=blue,
}
\definecolor{cite-blue}{RGB}{0,0,204}
\date{}
\makeatother

  \providecommand{\definitionname}{Definition}
  \providecommand{\lemmaname}{Lemma}
  \providecommand{\propositionname}{Proposition}
  \providecommand{\remarkname}{Remark}
\providecommand{\corollaryname}{Corollary}

\usepackage{bbold}
\newcommand{\Song}{Song\textsuperscript{$\mathbb{1}$} }
\newcommand{\Ming}{Ming\textsuperscript{$\mathbb{2}$} }
\newcommand{\Qing}{Qing\textsuperscript{$\mathbb{3}$} }

\bibpunct{(}{)}{;}{a}{,}{,}

\usepackage{tikz}
\usetikzlibrary{patterns}
\usetikzlibrary{datavisualization.formats.functions}
\usetikzlibrary{fit}
\usetikzlibrary{fadings}
\usetikzlibrary[graphs]
\usetikzlibrary{shapes.misc, positioning}

\begin{document}

\author{In\'{a}cio B\'{o}}
\author{Li Chen}

\address{\textbf{Bó}: Department of Economics, Faculty of Social Sciences, University of Macau, Macau; email:
\href{mailto:inaciobo@um.edu.mo}{inaciobo@um.edu.mo}.}

\address{\textbf{Chen}: Corresponding author. Shanghai College of Intellectual Property Rights, Tongji University, China; Department of Economics, University of Gothenburg, Gothenburg, Sweden; email: \href{mailto:chen\_li@tongji.edu.cn}{chen\_li@tongji.edu.cn}.}

\date{April, 2026}
\thanks{We thank Lars Böerner, Andrea Bréard, Rustamdjan Hakimov, Philipp Heller, Yuan Ju, Dongwoo Lee, Hervé Moulin, Luigi Pascali,  Marek Pycia, Nadja Stroh-Maraun, and Utku Ünver for helpful comments. Special thanks to Rui Magone for assistance in our historical research of the Chinese civil service system. Li Chen acknowledges financial support from the Fundamental Research Funds for the Central Universities (22120250312), the Anniversary Foundation of the School of Business, Economics and Law at University of Gothenburg, and the Tore Browaldh's foundation (BFh18-0007). Inácio Bó acknowledges financial support from the National Natural Science Foundation of China (Grant No. 72573203) and the Asia-Pacific Academy of Economics and Management seed grant.}

\title[Designing Heaven's Will]{Designing Heaven's Will:\\ The job assignment in the Chinese imperial civil service}

\maketitle
\begin{abstract}

We study the evolution of entry-level civil service job assignment in imperial China from the tenth to the early twentieth century. The procedures were sequential, irrevocable, and publicly verifiable, relying increasingly on lots-drawing and operating under constraints that limited which posts candidates could fill and where they could be posted. Using a unified matching framework, we compare procedures by their ability to maximize the number of filled posts and to assign higher-ranked jobs to higher-degree candidates. We show that reforms intended to improve these outcomes can instead reverse them once constraints interact, a failure driven by greedy sequential matching. Building on the structure of the final historical procedure, we describe minimal extensions of public lots-drawing mechanisms that guarantee maximum matchings under compatibility constraints.

\end{abstract}


\clearpage{}
\begin{flushright}
\par\end{flushright}

%%%%%%%%%%%%%%%%%%%%%% Introduction
\section{Introduction}

In this paper, we draw on historical documents and studies to present a formal description of one of the earliest and longest-used assignment schemes for allocating candidates to government jobs, in use from the mid-tenth century to the early twentieth century in China. We then develop a unified theoretical framework that clarifies the trade-offs underlying the successive changes to this system.

We identify two substantive criteria that constrained the operation of the Chinese civil service assignment procedures. The first is the \textit{eligibility} criterion, which restricted the set of government jobs to which a candidate could be assigned as a function of competence, as measured by examination degrees. The second is the \textit{rule of avoidance}, which prohibited candidates from being assigned to posts located in their home regions.

Assignments were carried out sequentially and irrevocably: candidates were considered one by one and matched to an available compatible job at their turn. This greedy structure made the procedure simple to implement and easy to monitor publicly, but also constrained the scope for ex post correction. From the late sixteenth century onward, transparency concerns were increasingly addressed by replacing discretionary job choice within this sequential process with publicly observable lots-drawing. Randomness was thus introduced not as a stand-alone objective, but as a means of ensuring that greedy assignments could be carried out in a transparent and verifiable way, while limiting personal influence and factional interference.

The following quote, from the court discussions between Emperor Chongzhen and his ministers in 1628, highlights the way in which the choice for resorting to randomness was framed as the delegation of these decisions to \emph{Heaven's Will}:
\smallskip 

\begin{quote}
``\emph{Finding the right men for the grand secretariat benefits
greatly the empire. I do not dare to make the decision myself, therefore, I ask for Heaven's will}.'' --- \cite{sun1777chunming}\footnote{Grand secretaries were cabinet members of the central administration. Lots-drawing was used mainly for assigning entry-level jobs which we will describe in detail later. It was however occasionally used for higher-level jobs.}
\end{quote}
\medskip 

The first formal assignment procedure appeared in the mid-tenth century, during the \Song dynasty (960–1279).\footnote{To assist readers in following the chronological order of the dynasties discussed in the paper, we add a superscript number next to each dynasty name: \Song precedes \Ming, which in turn precedes \Qing.} Candidates were assigned to jobs sequentially and publicly, in an order determined by degree and examination rank, with individual preferences guiding job choice. This sequential structure remained a defining feature of all subsequent procedures.

Major reforms followed during the \Ming dynasty (1368–1644).\footnote{The Yuan dynasty (1279–1368) interrupted the institutional continuity of the civil service system; although examinations were partially reinstated in 1315, systematic assignment procedures of the type studied here re-emerged only under the \Ming dynasty.}  From the late fourteenth century, the First \Ming procedure replaced candidate preferences with bureaucratic assignment, motivated by concerns that higher-degree candidates were being allocated to lower-level posts. In the late sixteenth century, the Second \Ming procedure introduced publicly observable lots-drawing, reducing discretionary interference but reintroducing mismatches between candidates and jobs.

Lots-drawing persisted into the \Qing dynasty (1644–1912), where reforms focused on improving its performance. In the mid-seventeenth century, the First \Qing procedure partitioned assignments by job type to better align qualifications and posts, and in the early nineteenth century the Second \Qing procedure further refined this system by prioritizing candidates constrained by the rule of avoidance, reducing the incidence of unfilled jobs.

To understand the trade-offs behind these changes, we introduce a formal model that captures the key elements of the different assignment procedures used in the period of time we considered. 
We focus on two measures when evaluating these procedures, \textit{minimizing the number of unfilled jobs}, and \textit{prioritizing high-levels}---favoring candidates with high-level degrees for high-level jobs. These two objectives were explicitly or implicitly mentioned in the historical context, which ensured that jobs were assigned, and to the candidates with the right competence. 
Both objectives were salient in the historical context and reflected concrete administrative concerns. Although the number of qualified candidates exceeded available positions throughout the \Song and later dynasties,\footnote{During the twelfth century, the ratio of candidates to vacancies was about 3 to 1 \citep{deng2004politics}.} eligibility and avoidance constraints could still leave posts unfilled. Unassigned candidates had to wait until the next appointment cycle—typically three months later—while positions remained vacant. Minimizing unfilled jobs was therefore a first-order concern, particularly when expanding the number of posts was difficult and ad hoc corrections would have undermined the transparency and rule-based nature of the assignment system.

%We present a unified theoretical framework for comparing the five assignment procedures, in terms of their ability to achieve the objectives of minimizing unfilled jobs and prioritizing high-levels. 
The theoretical framework allows us to compare the \Song, First, Second \Ming, First and Second \Qing procedures, despite the first involving preferences, the second arbitrary choices by bureaucrats, and the last three random draws. Our framework exploits the sequential nature of how all of the procedures are used in determining the final allocation, while being agnostic about the source of these orders. This can result in particularly robust conclusions, where comparisons of outcomes between random and deterministic procedures hold not in expectation, but in every realization of chance. As a result, it also allows us to isolate the effects of individual constraints, such as eligibility and rule of avoidance. 

When compatibility is defined solely by eligibility, the historical sequence of reforms largely moved in the intended direction. The transition from the \Song to the First \Ming procedure reduced the incidence of unfilled jobs and improved the assignment of high-level candidates to high-level posts. Although the introduction of lots-drawing in the Second \Ming temporarily weakened these properties, subsequent reforms under the \Qing dynasty restored them: both the First and Second \Qing procedures perform better along these dimensions while retaining public randomization (see Propositions~\ref{prop:CardinalEfficiencyAllProceduresCMinus} and~\ref{prop:AssortativenessAllProceduresCMinus}).

Once the rule of avoidance is taken into account in addition to eligibility, however, this pattern no longer holds. The same reforms—from the \Song to the First \Ming and from the Second \Ming to the First \Qing—can lead to more unfilled jobs and weaker prioritization of high-levels. These reversals arise from nontrivial interactions between eligibility and avoidance constraints, and show that reforms motivated by qualification-based considerations alone can produce worse outcomes once regional restrictions are respected (Theorem~\ref{prop:FromSongToMingOneLessAssortative}).

The final reform, from the First to the Second \Qing procedure, is qualitatively different. By giving priority to candidates constrained by avoidance before matching others, it weakly reduces the number of unfilled jobs in all cases and strictly reduces it in some, without undermining the assignment of high-level candidates to high-level posts (Theorem~\ref{prop:FromQingOneToQingTwo}). Historical evidence indicates lower average waiting times for appointment over this period, a pattern compatible with the increase in match cardinality implied by the reform.

Finally, we move from historical evaluation to mechanism design. Building on the structure of the Second \Qing procedure, we study how sequential lots-drawing mechanisms can be modified to improve match cardinality under compatibility constraints. We show that, for arbitrary compatibility structures, guaranteeing maximum matchings through sequential drawing may require an arbitrarily large number of job tubes, so that no fixed, simple tube architecture can be universally sufficient (Proposition~\ref{prop:unboundedNumberOfTubes}).

We then identify a sharp positive result for a natural and historically relevant class of environments. When compatibility is driven by the rule of avoidance and the market satisfies a mild regularity condition, a simple two-by-two structure of worker and job tubes guarantees a maximum matching for every realization of chance (Theorem~\ref{thm:TwotubesOnEachSide}). This construction preserves transparency and public verifiability, while extending the historical lots-drawing procedure in a minimal way. Indeed, we show that this two-by-two structure is minimal: no procedure with fewer worker or job tubes can ensure maximality even in regular environments (Proposition~\ref{prop:twotwominimal_corrected_final_v4}).

Finally, we describe the limits of transparent random design once additional objectives are imposed. When eligibility constraints and the prioritization of high-level posts are considered alongside avoidance, no anonymous lots-drawing procedure can simultaneously guarantee maximum matchings and prioritization (Proposition~\ref{prop:incompatibilityMaximizingAndPrioritizing}). Together, these results delineate the scope and limits of simple, transparent random assignment mechanisms that follow the structure introduced by our historical analysis.

Beyond its historical motivation, however, our analysis speaks to contemporary applications in which public randomization is valued, such as allocating refugees to host communities with capacity constraints \citep{andersson2017assigning}, assigning judges with heterogeneous expertise to cases \citep{thorley2020randomness}, or allocating public housing with size constraints \citep{arnosti2020design}. Lots-drawing has also been proposed as a governance tool, for example in distributing positions within the European Commission among member states \citep{buchstein2009randomizing,berger2020prevent}.

%%%%%%%%%%%%%%%%%%%%%% Related Literature
\subsection{Related Literature}

Given its long history and prominent role in government and society, the workings of the Chinese civil service have been the subject of an extensive bibliography, especially among historians. In the context of the Chinese civil service, \citet{will2002creation} is the first historical study that reviews the origin and evolution of the lots-drawing procedure in imperial China since the late sixteenth century. \citet{watt1972district} offers a panoramic view of the career path for county magistrates---from first-time appointment, promotion, and re-appointment to demotion---during the period between the late eighteenth and the early nineteenth century. Similarly, \cite{gong1997dictionary}, \cite{pan2005ming} and \cite{zhang2010qing} review the civil service systems for \Song, \Ming and \Qing respectively. 
These historical studies mainly analyze the societal and political background of the appointment systems, but have not evaluated the properties of the appointment procedures.  
A few recent economics papers evaluate empirically how the appointment methods, and in particular patronage---by means of discretionary appointments through connections, contrary to our focus of rule-based assignment---affect public service, and find that patronage generally leads to worse economic performance or the selection of less competent officials \citep{xu2017costs,colonnelli2020patronage}. \cite{thakur2020indian} studies a formal assignment procedure currently used in the Indian Administrative Service, focusing its distributional issues and the effect on bureaucratic and economic performance.    

By formalizing the civil service assignment as a matching problem, our paper connects to the literature in matching and market design, which has a long tradition of applying formal economic theory to study properties of allocation mechanisms used in real-life applications \citep{roth1984evolution,sonmez2013bidding,sonmez2013matching,dimakopoulos2019matching,abdulkadirouglu2003school,roth2004kidney}. Some of the issues that we identify regarding the way in which the matchings are produced sequentially are indirectly related to, for example, the topic of reserve design. \cite{dur2018reserve} show how the processing order of reserves in school choice design matters and how some processing orders can cause unintended consequences. Other papers include \cite{sonmez2021affirmative} on affirmative action policies in Indian civil service, \cite{pathak2020immigration} on H-1B visa allocation rules, and \cite{pathak2020fair} on medical resources allocations. In addition, the search for procedures that guarantee the minimization of unfilled jobs relates to maximum matchings in random environments. \cite{bogomolnaia2004random} look at maximum matchings when randomizing over dichotomous preference. \cite{boczon2018goals} study the UEFA Champions League group drawing method, a problem in which matchings are also determined by drawing lots and where there is a concern about maximality. 

More broadly speaking, our paper also relates to a burgeoning literature that applies economic theory to analyze important economic and societal institutions from a long-run perspective. \cite{greif1993contract} is perhaps the first paper in the field, applying contract theory to analyze contractual relations in eleventh-century Mediterranean trade. Other studies include  \cite{greif1994coordination} on the contract enforcement problem faced by merchants in late medieval Europe; \cite{borner2017design} on the debt-clearing financial mechanisms in preindustrial Europe; \cite{mackenzie2019axiomatic} on the succession rules of Popes; and \cite{boerner2022medieval} on the medieval brokerage mechanisms in wholesale markets.  
Our paper follows this approach in that we use the toolbox from matching theory to analyze the civil service assignment procedures---another important historical institution---in order to shed light on its properties. 
 
The remainder of the paper is organized as follows. Section \ref{sec:Background} provides the historical and institutional background on the selection and appointment of civil servants, including the examination system, job structure, and the key constraints—eligibility, avoidance, and public randomization—that shaped assignment practices, and it reconstructs the assignment procedures used from the \Song to the \Qing dynasties. Section \ref{sec:theoreticalframework} develops a unified theoretical framework that formalizes these procedures within a common matching model. Section \ref{sec:AnalyzingTheProcedures} analyzes the outcomes generated by the historical procedures, establishing comparative results on match cardinality and the prioritization of high-level posts, and highlighting how reforms can interact non-trivially with constraints. Section \ref{sec:NewMechanisms} moves from historical evaluation to mechanism design, discussing the limits of sequential lots-drawing and introducing transparent mechanisms that guarantee maximum matchings under avoidance constraints. Section \ref{sec:conclusion} concludes and discusses broader implications and directions for future theoretical and empirical work.\footnote{Additional historical background, source documentation, and supplementary material are provided in an online appendix.}

\section{Selection and appointment of civil servants:\\A historical background \label{sec:Background}}

\subsection{Institutional setting and examination-based qualification}

Evidence of state bureaucracy in China can be traced back to as early as the third century BC \citep{creel1964beginnings}. A large body of bureaucrats was needed to implement various tasks, ranging from tax collection and juridical investigation to miscellaneous administrative chores. 
The way in which these bureaucrats were selected has evolved over time, starting from a hereditary system to a recommendation-based system in 134 AD, and finally to an examination-based system in 589 AD, known as the \textit{Civil Service Exams}. It lasted until 1905,\footnote{The exam system was interrupted during the Yuan dynasty (1279--1368), a Mongol-led dynasty. It was re-established in 1315, not long before the fall of the Yuan dynasty.}  and was abolished following a series of modernization reforms shortly before the fall of Imperial China.
The Civil Service Exams allowed any man to register for the exams, without recommendation or patronage by incumbent officials, and registered candidates were selected through a series of standardized exams that primarily focused on Confucian classics.
It was hoped that these exams would ensure impartial evaluation and select candidates by merit rather than birth or class.
Selection by exams was initially used on a small scale for the civil service,\footnote{Besides examination, one could qualify for the civil service through other channels. This includes clerks who were hired by local governments on term-limited contracts, people who were qualified in recognition of distinguished services by their fathers, and those who were qualified by paying a tribute to the government. %(\emph{juan na}). 
These candidates needed to pass separate exams before being considered for appointments. }
but started to gain importance in the \Song dynasty (960--1279). 
The \Song administration gradually established a career civil service system, where  
candidates were first selected by merit, and once qualified they could follow a career ladder for promotion and remain on the payroll of the central government until retirement, even if no specific tasks or posts were available \citep{gong1997dictionary,deng2004politics}.
In addition, rules were introduced to safeguard impartiality in exams, which remained in the following dynasties. 
For instance, names of candidates were obfuscated, and their exam answers were transcribed before evaluation, 
all to make sure that the candidates were not able to be identified directly or indirectly through handwriting \citep{elman2000cultural}. Selection by exams gained significant importance in the \Ming dynasty (1368--1644) and the succeeding \Qing dynasty (1644--1912), and was regarded as the ``regular'' path to qualify for civil service, in contrast to other paths of qualifications.\footnote{These candidates were intended for jobs overseen by the central government, both within and outside the metropolitan area. Local governments were led by people appointed by the central government, along with  
a significant number of subordinate officials who qualified through alternative channels, such as purchased degree holders \citep{Wu2013Juanna}. 
%In contrast, the percentage of
}  
In the remainder of the paper, we focus on the appointments of candidates who obtained their qualifications through examination. 

\subsection{Degrees, appointment cycles, and entry-level posts}

The Civil Service Exams included multiple levels of qualification exams. After passing the entry-level exam, held locally and twice every three years, candidates were given the degree of \emph{licentiate} (``\emph{xiu cai}''). Licentiates could then proceed to the next levels of triennial exams which offered opportunity for centrally appointed government jobs.  
The second level was a provincial-level examination in the provincial capital, and candidates who passed these were awarded the degree of \emph{recommended men} (``\emph{ju ren}''). Recommended men could proceed to the third-level exam held in the capital. Successful candidates were awarded the provisory title of \emph{tribute scholars} (``\emph{gong shi}''). They were re-examined and ranked in the imperial palace, under the supervision of ministers and often the emperor himself. These candidates were eventually awarded the degree of \emph{advanced scholars} (``\emph{jin shi}''). 
During a typical year of examinations in \Ming and \Qing dynasties, more than 1,000 candidates received the degree of recommended men, while about 100 went on to obtain the degree of advanced scholar. 

Recommended men, tribute and advanced scholars were then appointed by the Ministry of Personnel to various important government jobs located in or outside the capital. 
These jobs were entry-level civil service jobs, while higher-level jobs were reserved for the promotion of more experienced officers. Once assigned, the new officials would serve three years in their positions, and after that, depending on their evaluations, they could be promoted, transferred, or demoted. 
Typical entry-level jobs in the capital included editors at Hanlin Academy---an elite scholastic institute whose main task was to interpret the Confucian classics---secretaries in various ministries, officers at various departments and courts, etc. 
Typical jobs outside the capital included magistrates and judges in prefectures, sub-prefectures, and counties.\footnote{From the \Song to the \Qing dynasty, the central administrative divisions were such that outside the capital city, province was the primary division. Under the level of province were prefectures, followed by sub-prefectures. The smallest administrative unit was county. Each level of the administrative units was typically headed by a magistrate, who was assisted by a judge in charge of investigating civil and criminal cases.} Among all jobs, the most common ones were county magistrates, who were responsible for the overall management of a county, including mostly but not exclusively tax collection, law enforcement, school inspection, and disaster relief. There were roughly 1,300 to 1,500 counties during our periods of interest, and most of the open posts were reserved for the newly selected candidates, with some remaining for transfer or re-appointment posts. For instance, about two-thirds of the posts were for newly selected candidates during the early \Qing period \citep{will2002creation}. 

\subsection{Constraints and desiderata}

The overall stated goal of the appointment was to ``find the right person for the right job.'' There are of course different interpretations on what should be considered an appropriate match. We observe three criteria that emerged over time, and were respected by the appointment system, though to different extents. While the first two criteria were present in all procedures, the third one appeared as of the late sixteenth century.  

\subsubsection*{Eligibility.}

Qualified candidates, depending on their degree, were eligible for different sets of civil service jobs. In general, candidates with higher degrees were appointed to  jobs with higher ranks. That is, advanced scholars were more likely to be assigned to more important jobs than tribute scholars or recommended men. Early debates often centered around whether academic achievement should be the sole consideration for measuring one's competence. However, over time,  academic achievement was gradually accepted as the main desideratum, as it provided clear incentives to follow the career system and gave legitimacy to the examination system. 
Since the precise eligibility requirement varied with time, we will elaborate on the details in Section \ref{sec:AssignmentsThroughHistory} when discussing the assignment procedures. 

\subsubsection*{Rule of Avoidance.}
This additional criterion prevented a candidate from being assigned to a particular job even if he was eligible. The rule of avoidance, in fact much older than the civil service appointment, dates back to the second century. Despite the various forms it took over time, avoidance of localities is the most fundamental one for the entry-level jobs \citep{Wei1992}.
It stated that a candidate was prohibited from being appointed to his native province.\footnote{The other types of avoidance can overlap with avoidance of localities or involve smaller sets of people. For instance, the avoidance of family, which prevented a candidate from being assigned to a job where an incumbent official were direct or indirect family members. Since most people’s families came from the same region, avoidance of localities was often sufficient. The other type was avoidance of teacher and student, which prevented a candidate from serving in a job where he was the student or the teacher of an incumbent officer.} 
Bureaucrats originating from a particular region naturally had an information advantage about their home regions, and therefore assigning them there could be beneficial to the local management. However, as they were also more connected to local elites, this also posed a threat of local capture, especially when regions were more isolated due to communication and transportation constraints. 
Rule of avoidance was thus intended to prevent the formation of local powers, which was regarded as a major impediment and challenge to unification throughout the Chinese history. 

\subsubsection*{Randomness and public verifiability.}
Starting from the late \Ming dynasty, candidates were assigned to compatible posts through publicly observable lots-drawing, in which job titles were placed into tubes and drawn at random. Though it might seem odd to randomly assign civil servants to posts, it is a way by which the emperor could essentially mitigate the ability of local clans extending their power by influencing these assignments, without taking on the decisions himself. On the other hand, the lack of criteria when deciding individual assignments reduces the ability of matching highly qualified candidates to higher-level jobs. It seems, however, that randomness was gradually accepted as a desirable property from the late sixteenth century onward.\footnote{Random assignments have a surprising history in the political sphere. As early as the democratic period of ancient Athens, random assignments were used to select citizens to serve in the Boule (a council appointed to run the daily affairs of the city) and various state offices, through a randomization device known as Kleroterion \citep{headlam1891election}. In recent years, political scientists have also advocated for a lot-drawing procedure for assigning EU commissioners \citep{buchstein2009randomizing}.}

\subsection{Sources, scope, and reconstruction of procedures}

We now turn from general institutional background to a detailed description of the concrete assignment procedures used from the tenth to the early twentieth century. We focus on the methods that were used from the tenth to the early twentieth century to determine, for each of the candidates selected by exams, which jobs they were assigned to. In order to obtain an overview on how the assignment procedures worked, we combine both official documents as well as secondary sources. Official documents for dynasties before the tenth century do not inform us about the details of the appointment procedure, perhaps due to its small scale, and assignments were more ad-hoc. The earliest official regulations describing the assignment procedures in detail are from the \Song dynasty.\footnote{Note that official documents were often compiled in the succeeding dynasty by the order of the new emperors. Writing ``official dynastic history'' of its predecessor is a convention established by the state historian Sima Qian from Han dynasty (202BC–-220), which, even though it might be biased, provides invaluable sources to study institutions in earlier times. Parts of the original official documents from the \Ming dynasty and the majority of original official documents from the \Qing dynasty are preserved today.} The treatises on selection and appointment from the History of Song \citep{Song1343}---the official dynastic history---describe the types of civil service jobs that candidates were assigned to, which draws a blueprint for the dynasties to follow, though details are often vague. Rules of the Ministry of Personnel from Yongle Encyclopedia \citep{Yongle} offer supplementary materials for understanding how the assignment procedure was actually carried out during this period. For the following \Ming dynasty, in addition to treatises on selection and appointment from the History of Ming \citep{Mingshi}, Collected Statutes of Ming \citep{Minghuidian}---the official code of administration from the \Ming dynasty documenting regulations for six ministries, including the Ministry of Personnel, and Regulations of Ministry of Personnel \citep{zhizhang}---both provide additional descriptions of the types of jobs candidates were eligible for and the assignment procedure. Lastly, with the maturing of the assignment system, the \Qing dynasty produced the most detailed official sources on the assignment procedure, including how it was carried out and how eligibility and rule of avoidance were respected in the procedure, as well as related policy changes (see Collected Statutes of Qing \citep{daqingGuangxu} and Regulations of Ministry of Personnel \citep{zeli1886}). Our secondary sources include court debates, correspondences between emperors and ministers, and handbooks preparing students for the career system \citep{huang1997fuhui}. These sources allow us to validate implementation details mentioned in the official documents. 

\subsection{Assignment procedures through history}
\label{sec:AssignmentsThroughHistory}

\subsubsection{The \Song procedure (960–1279)}
\label{subsec:SongProcedure}

With the expansion of recruitment through civil service examinations, a formal assignment procedure emerged during the \Song dynasty. Successful candidates were eligible for entry-level civil service positions, including clerical posts in the capital and county magistracies in the provinces, although eligibility rules were not sharply delineated by degree category.

Appointments were conducted quarterly through a priority-based, sequential procedure. Candidates were ordered primarily by degree—advanced scholars first—and then by examination results within each degree category.\footnote{In addition to exam performance, endorsements by incumbent officials could affect priority during the \Song dynasty; this hybrid merit–patronage feature was abandoned in later procedures.} All vacancies were publicly announced in advance. Following the priority order, candidates were called one by one and indicated a preferred available job, which officials would either approve or reject. Unassigned candidates and jobs were reconsidered in an additional round, with any remaining unmatched positions carried over to the next appointment session three months later.

\subsubsection{The First \Ming procedure (1368–1594)}
\label{subsec:MingOneProcedure}

The First \Ming procedure retained a priority-based, sequential assignment, but eliminated candidates’ job choice. Instead, the Ministry of Personnel evaluated candidates and assigned them to suitable positions. Appointments were organized monthly—first-time appointments in even months and promotions or transfers in odd months—a structure that persisted in later dynasties.

Candidates were ordered by degree and examination results: advanced scholars first, followed by tribute scholars and recommended men. Using this priority order, the ministry assigned candidates to eligible jobs while respecting the rule of avoidance. Assignments were intended to match candidates’ skills as measured by exam performance, so higher-ranked candidates were more likely to receive higher-ranked posts.

A key innovation of the First \Ming procedure was the formalization of eligibility rules. As summarized in Table~\ref{tab:ming}, advanced scholars were divided into three grades with distinct eligibility sets, while tribute scholars and recommended men were eligible only for non-metropolitan positions. Some jobs were reserved exclusively for advanced scholars, others only for lower-degree candidates, and some could be filled by candidates from any degree category. These rules sharply constrained the feasible assignments relative to the \Song procedure.

While the procedure could in principle accommodate both eligibility and avoidance constraints, it left substantial discretion to officials. By the late sixteenth century, this discretion was widely perceived to enable factional and clan-based interference, undermining meritocratic assignment and motivating subsequent institutional reforms \citep{will2002creation}. Unlike the other procedures we study, the First \Ming procedure does not fully specify how the ministry selects among eligible jobs; in our theoretical analysis, we therefore adopt the optimistic assumption that assignments favor better skill–job matches.

\begin{table}[htp]
    \centering
    \small
    \caption{Jobs eligible to candidates by degree in the Ming dynasty}
    \label{tab:ming}
    \begin{tabular}{p{3.5cm}p{3cm}p{4cm}p{4cm}}
        \toprule
        & & \multicolumn{2}{c}{Jobs} \\
        \cmidrule{3-4}
        & & Metropolitan & Non-metropolitan \\
        \midrule
        \multirow{3}{*}{\parbox{3.5cm}{Advanced scholars}} & First grade & Senior and junior editors at Hanlin Academy & \\
        \cmidrule(l){2-4}
        & Second grade & Ministry secretaries & Sub-prefecture magistrates \\
        \cmidrule(l){2-4}
        & Third grade & Secretaries to the emperor, officers at Department of Foreign Affairs, officers at Court of Judicature Review, officers at Court of Ceremonials & Prefecture judges, and county magistrates \\
        \midrule
        \parbox[t]{3.5cm}{Tribute scholars and\\recommended men} & & & Prefecture judges, county magistrates, and study officers \\
        \bottomrule
    \end{tabular}
\end{table}

\subsubsection{The Second \Ming procedure (1594–1644)}
\label{subsec:MingTwoProcedure}

In 1594, the \Ming administration introduced a lots-drawing procedure, marking the first systematic use of lotteries to assign newly selected candidates to entry-level civil service jobs. Eligibility information was publicly announced in advance. Job titles were written on sticks and placed into tubes, from which candidates drew their assignments.

The procedure was sequential and priority-based. Advanced scholars drew first, followed by tribute scholars, and then recommended men, with candidates ordered within each group by examination results.\footnote{While official documents are vague about the drawing order, secondary sources indicate that jobs remaining after advanced scholars were drawn were carried over to tribute scholars and recommended men \citep{Wu1609}.} Each candidate drew from a tube containing only jobs for which he was eligible. Although the sources are imprecise about implementation, both context and later \Qing descriptions suggest that draws violating the rule of avoidance were discarded and redrawn until a compatible job was found.

Relative to the First \Ming procedure, lots-drawing substantially reduced discretion and personal interference by introducing randomness into the assignment process. This was explicitly intended to limit factional influence and regional favoritism, drawing on an established cultural practice already used in other administrative contexts \citep{will2002creation}. At the same time, the procedure was criticized for its lack of control over matching quality: high-priority candidates could be randomly assigned to mid-ranked jobs while higher-ranked positions remained vacant, and reports of manipulation of tube contents indicate that corruption quickly re-emerged \citep{Guyanwu1670,shen1619wanli}.

\subsubsection{The First \Qing procedure (1644–1824)}
\label{subsec:QingOneProcedure}

The lots-drawing procedure was retained after the dynastic transition, as random assignment was viewed as a tool to limit regional power and reinforce central control under the new Manchu-led administration. Eligibility rules were slightly revised, and additional safeguards were introduced to enhance impartiality: appointments were conducted publicly under the supervision of the Censorate, candidate and job names were sealed before drawing, and drawing order was itself randomized rather than determined by exam rank.

The main institutional innovation was the introduction of a \emph{partitioned assignment} system. Vacancies were first grouped by job type (e.g., ministry secretaries, county magistrates). For each job type, the Ministry of Personnel constructed a corresponding list of eligible candidates, selecting candidates by exam rank up to the number of vacancies. When multiple degree categories were eligible for a given job type, a quota system determined the composition of candidates placed into the tube.\footnote{This quota system, known as the ``monthly class'' (\textit{yueban xuance}), functioned analogously to a reserve system: it restricted eligibility, while the final assignment within each job type was still determined by lots.}

Assignments were then carried out independently for each job type. For a given type, a supervisor sequentially drew a candidate and a job from their respective tubes. Matches respecting the rule of avoidance were immediately announced; otherwise, jobs violating avoidance were set aside and redrawn until a compatible match was found, after which excluded jobs were returned to the tube. The process continued until all candidates were matched or no compatible pairs remained. Unassigned candidates and jobs were carried over to the next appointment cycle.

By separating high-ranked jobs into distinct assignment markets, the First \Qing procedure reduced the likelihood that high-priority candidates would be assigned to lower-ranked jobs while higher-ranked positions remained available, a key criticism of the Second \Ming procedure. At the same time, partitioning narrowed the set of jobs available to each candidate, increasing the risk that avoidance constraints would leave some jobs unfilled. This risk was further amplified by the rule that candidates unmatched within a job type could not be considered for other types in the same round, a restriction intended to preserve rank-based assignment across appointment cycles.

\subsubsection{The Second \Qing procedure (1824–1905)}
\label{subsec:QingTwoProcedure}

In 1824, the \Qing administration modified the partitioned lots-drawing procedure to address the increased incidence of unfilled jobs generated by partitioned assignments. The key change was to prioritize, within each job type, candidates facing incompatibilities due to the rule of avoidance: candidates who needed to avoid some jobs in the tube were required to draw first, and only after they completed their draws could the remaining candidates proceed. All other aspects of the procedure remained unchanged. This reform is explicitly recorded in the Collected Statutes of \Qing \citep{daqingGuangxu}:
\medskip

\begin{quote}
``\emph{1824, it was approved after discussions, for the people who draw lots in the monthly appointment, those who have home provinces to avoid draw first. If they still draw a job that needs to be avoided, remove this job and ask {[}the candidates{]} to draw another job. Until a {[}compatible{]} lot is drawn, let those who do not need to avoid home provinces draw.}''
\end{quote}
\medskip

By prioritizing candidates constrained by avoidance, the reform increased their likelihood of finding a compatible job and weakly improved match cardinality. Since assignments remained partitioned by job type, avoidance was the sole source of incompatibility within each assignment market. While we lack direct data on match cardinality, archival evidence suggests a reduction in waiting times following the reform: \citet{Wang2016} reports that the average time for advanced scholars to become county magistrates fell from about eight years (1796–1820) to 5.5 years (1821–1850).

Overall, the historical evolution reflects a gradual shift toward more systematic and transparent assignment rules. Candidate preferences, present in the \Song procedure, were abandoned as incompatible with rank-based priorities. From the late sixteenth century onward, lots-drawing became routinized, with successive refinements aimed at balancing eligibility, avoidance, randomness, and match cardinality. Across all procedures, assignments were conducted publicly and sequentially, with matches announced as soon as they were determined—ensuring transparency and feasibility given the technological constraints of the period.

\section{Assignment procedures: A theoretical framework}\label{sec:theoreticalframework}
Analyzing the functioning of the assignment procedures used in several historical episodes is, of course, challenging, as the documentation becomes sparser further back in time. Our analyses, therefore, are the result of our best effort to derive coherent procedures that were consistent with the descriptions, related discussions, and the descriptions of the procedures that replaced them. It should surprise no one that there are, also, recorded instances of ad-hoc adjustments on how they were implemented.\footnote{For instance, it was documented that in the \Qing dynasty, candidates who were initially appointed through the standard procedure had to present themselves in the court and in front of the emperor for audition \citep{zeli1886}. Based on the audition and the difficulty level of the job one was assigned to, the initial assignment could be changed.} We abstract away from such  adjustments and construct a simple model to compare the main assignment procedures across time. While motivated by the civil service assignment, our theoretical framework follows the terminology used in matching literature. The terms \textit{worker} and \textit{candidate} are used interchangeably. 

\subsection{Model}

A set of \textbf{workers} $W$ is to be matched to a set of \textbf{jobs} $J$. A \textbf{matching} is a function $\mu:W\cup J\to W \cup J \cup \{\emptyset\}$ such that each worker (job) is assigned either to one job (worker) or is left unassigned, and a worker is matched to a job if and only if the job is also matched to him, that is, for any $w\in W$ and $j\in J$,  $\mu(w)=j\iff \mu(j)=w$. To account for the restrictions on what matches are acceptable, let $\mathcal{C}$ be a \textbf{compatibility} correspondence $\mathcal{C}:W\twoheadrightarrow J$ that defines which jobs are compatible with each worker. In our application, compatibility could be determined by eligibility or the rule of avoidance.  
In addition, we say that a matching $\mu$ is \textbf{feasible} if workers only receive jobs that they are compatible with, that is, for any $w\in W$ and $j \in J$, $\mu(w)=j\implies j\in\mathcal{C}(w)$. The set of all feasible matchings is denoted by $\mathcal{M}$.
We refer to a triple $\langle W,J,\mathcal{C}\rangle$ as a \textbf{market}.

For any matching $\mu \in \mathcal{M}$ and any subset of jobs $J'\subseteq J$, let $\mu(J')$ be the set of workers matched to some job in $J'$ under matching $\mu$. That is, $\mu(J') \equiv \left\{w\in W: \mu(w)\in J'\right\}$. 
Denote by $|\mu|$ the cardinality or the size of $\mu$, that is, $|\mu|=|\mu(J)|$. 
\begin{defn}
A matching $\mu$ \textbf{minimizes unfilled jobs} if for any $\mu' \in \mathcal{M}$, $|\mu|\geq |\mu'|$. 
\end{defn}

Our framework for evaluating the mechanisms is based on two elements: \textit{assignment plan} and \textit{assignment arrangement}. They are, loosely speaking, respectively analogous to an \textit{economy} and a \textit{mechanism} in a typical market design model. In a Chinese civil service assignment problem, elements of an economy can be the preferences of workers (in the \Song procedure), the realization of chance in the draws from the tube of jobs and workers (in the Second \Ming and both \Qing procedures) or considerations made by the Ministry of Personnel (in the First \Ming procedure). On the other hand, the different procedures share a ``sequential consideration'' aspect that can be captured by a simple and also more general structure, which we denote by assignment arrangements. They define sequences of subsets of workers and jobs to be considered for matching in a way that combines with assignment plans to result in a specific matching. One of the key advantages of this framework is that it allows us to perform comparative statics fixing a mechanism while varying the economies, but also fixing an economy and varying the mechanism. 

Formally, an \textbf{assignment plan} is a pair $\left\langle \succ^W,\left(\succ_w^J\right)_{w\in W}\right\rangle$, where $\succ^W$ is a strict total order over the set of workers, and each element of $\left(\succ_w^J\right)_{w\in W}$ is a strict total order over the set of jobs. The first element of the assignment plan, the order $\succ^W$, tells for each pair of workers $w$,$w'$, who will be considered for a matching first. In the \Song, First, and Second \Ming procedures, for instance, an assignment plan in which $w\succ^W w'$ represents the situation in which worker $w$ obtained a higher exam grade than $w'$, and therefore if both obtained the same title (for example, advanced scholar), both procedures will match $w$ to a job (if any) before $w'$.\footnote{Note, however, that $w\succ^W w'$ does not say that $w'$ will be matched immediately after $w$, but that $w'$ will not be considered for a match before $w$.} In the \Qing procedures, $\succ^W$ indicates the order in which workers in the same tube are drawn. So if, for example, a tube contains workers $w,w',w''$, then $w\succ^W w'\succ^W w''$ represents the realization of chance in which these workers are drawn from that tube in that order. All possible orderings in which workers can be drawn from that tube are represented by all the permutations of $\succ^W$. Importantly, a uniform distribution over all of these permutations represents the distribution in which these orders take place in the real-life procedure.

The second element of the assignment plan, the orders over jobs $\left(\succ_w^J\right)_{w\in W}$, represents the order in which jobs are considered, for each worker. Given a worker $w$ and two jobs $j$ and $j'$, $j\succ_w^J j'$ says that, when considering a match for worker $w$, if both jobs are being considered and the worker is matched to $j'$, it must be that $j$ is not available anymore or is incompatible with $w$. In the \Song procedure,  $\succ_w^J$ represents worker $w$'s preference over the jobs: if both $j$ and $j'$ are still available when his turn comes, and $j$ is compatible with him, he will choose $j$ instead of $j'$. In the First \Ming procedure, it represents whatever criterion was used by officials to choose among eligible jobs. Finally, in the Second \Ming and the \Qing procedures, orders over jobs also represent a realization of chance of the order of drawing of jobs from the corresponding tube. So if worker $w$ is such that $j\succ_w^J j' \succ_w^J j''$, and these three jobs are in the tube from which he will draw, he will first draw $j$. If $j$ is an eligible job that does not violate the rule of avoidance, he will be matched to that job. If, say, $j$ is incompatible with $w$, then he will next draw $j'$, and so on. As in the order over workers, drawing uniformly from the possible permutations over a set of jobs results in the same distribution of outcomes that result from the real-life procedure when $w$ is facing a tube with these contents. Note that, as opposed to the order over workers, which is unique, here each worker might be associated with a different order over jobs. This represents the fact that the drawings of jobs from a tube are independent draws---conditional, of course, on the contents of the tube. 

The other key concept when modeling the procedures is an \textbf{assignment arrangement}. An assignment arrangement is a list $\Phi=\left\langle\varphi^1,\ldots,\varphi^{\ell}\right\rangle$ that partitions a market into $\ell$ independent sub-markets. Each element of the list $\varphi^i$ is a \textbf{sequence of tubes}, consisting of a list of subsets of workers $\left\langle W_1^i, W_2^i, \ldots, W_{n_i}^i \right\rangle$ and a list of subsets of jobs $\left\langle J_1^i, J_2^i, \ldots, J_{m_i}^i \right\rangle$. 
These sets of workers and jobs, which we denote by \textbf{tubes}, are such that:
 \[
  \bigcup_{j=1,\ldots,\ell}\bigcup_{i=1,\ldots,n_j} W_i^j=W\text{    and    } \bigcup_{j=1,\ldots,\ell}\bigcup_{i=1,\ldots,m_j}J_i^j=J,
 \]
and for every $a\neq c$ or $b\neq d$, $W_a^b\cap W_c^d=\emptyset$ and $J_a^b\cap J_c^d=\emptyset$. In other words, the sets listed in the sequences of tubes partition the sets of workers and jobs.

An assignment arrangement represents two aspects of the assignment procedure. First, it allows for the problem to be partitioned into independent sub-markets. This is what is done in the \Qing procedures: candidates and jobs are split into separate and independent matching sub-problems for each type of job. Second, it represents the order in which different subsets of candidates and jobs are considered. A sequence of tubes $\left\{\left\langle W_1, W_2\right\rangle,\left\langle J_1, J_2\right\rangle\right\}$, for example, indicates that first the workers in $W_1$ will be considered and matched to jobs. Only after all the workers in $W_1$ are matched to a job (or are left without any remaining compatible job), the workers in $W_2$ will be considered. When it comes to jobs, take any worker (regardless of whether it is a worker in $W_1$ or $W_2$), jobs in $J_2$ will only be considered if there are no jobs compatible with that worker among those remaining in $J_1$. This representation not only allows the modeling of all the procedures, but also the extension to a new procedure that we introduce in Section \ref{sec:NewMechanisms}.

Given a market, an assignment plan together with an assignment arrangement can be combined to represent an execution of the matching procedures that we model. Consider a market $\langle W,J,\mathcal{C}\rangle$ and an assignment plan $\left\langle \succ^W,\left(\succ_w^J\right)_{w\in W}\right\rangle$. We next describe how the assignment arrangement produces a matching of workers to jobs. To do so, it is helpful to use the following notation. Given a set $I\subseteq W\cup J$ and an order $\succ$, we denote $top^{\succ}(I)$ as the top element of $I$ in the order $\succ$, that is, $top^{\succ}(I)=\left\{i\in I|\forall i'\in I\backslash \{i\}: i\succ i'\right\}$.

For each element $\varphi^i$ in the assignment arrangement, the steps below are followed to produce a matching $\mu$:

\begin{enumerate}
    \item Let $J^*=\cup_{k=1}^{m_i}J_{k}^{i}$.
    \item For each $t=1,\ldots,n$, let $A^t=W_t^i$ and repeat the following procedure until $A^t=\emptyset$: 
    \begin{itemize}
        \item Let $w=top^{\succ^W}(A^t)$ and remove $w$ from $A^t$. There are two cases:
        \begin{itemize}
            \item $\mathcal{C}(w)\cap J^{*}=\emptyset$: let $\mu(w)=\emptyset$.
            \item $\mathcal{C}(w)\cap J^{*}\neq\emptyset$: let $a^*$ be the lowest value of $a$ such that $\mathcal{C}(w)\cap J^{*}\cap J^t_a\neq\emptyset$, and $j=top^{\succ_w^J}\left(\mathcal{C}(w)\cap J^{*}\cap J^t_{a^*}\right)$. Let $\mu(w)=j$, and remove $j$ from $J^{*}$.
        \end{itemize}
    \end{itemize}
\end{enumerate}

In other words, given $\varphi^i$ in the assignment arrangement, we follow the order of tubes of workers in $\varphi^i$, matching the workers in $W^i_1$ first, and only after considering all of them do we move to the next subset of workers $W^i_2$, and so on. Within each tube, the order in which workers are chosen is determined by $\succ^W$. Each worker $w$ is matched to a job in the first tube of jobs that still contains a compatible job. If that tube contains more than one compatible job, then he is matched to the job with highest priority in that tube with respect to $\succ^J_w$. If, however, there is no compatible job left, then he is left unmatched. We repeat these steps for each $\varphi^i$ in the assignment arrangement, which produces a matching of workers to jobs. 
We refer to the combination of an assignment arrangement with a method of producing an assignment plan as a \textbf{procedure}.

The key advantage of this model, for our purposes, is that it allows us to perform reliable comparative statics between procedures. By fixing an assignment plan and evaluating the matchings produced by two different assignment arrangements, we are able to compare, for example, two random procedures without picking different realizations of chance. Moreover, we are even able to compare deterministic procedures, such as the \Song, to a random one, such as the Second \Ming. If we are able to analyze the outcomes produced by these two procedures for any fixed assignment plan, then the fact that one is deterministic and the other is random is inconsequential for the analysis.

To facilitate our analysis of the procedures used in history, we consider in what follows a simplified setting where there are two categories of workers, $W^A$ (representing advanced scholars), and $W^B$ (representing tribute scholars and recommended men). This simplification constitutes a common denominator, across all dynasties that we consider, in the configuration of jobs and candidates. While there was an extra division within the category of advanced scholars during the \Song and \Ming dynasties, our simplified version is equivalent, for our analysis, to the one used during the \Qing dynasty.
On the other side, jobs are partitioned according to eligibility: jobs in $J^A$ and $J^B$ are eligible for workers in $W^A$ and $W^B$ respectively, and jobs in $J^{AB}$ are eligible for all workers. 
The partition reflects the fact that, in our application, some jobs were assigned only to advanced scholars, and some only to tribute scholars and recommended men, while some others could be matched to candidates from both categories. 
This simplification allows us to capture the main features of the civil service assignment without unnecessary complications that are not crucial to our results. Moreover, there is a finite set of regions $\mathcal{R}$. Each worker and each job is associated to a region in  $\mathcal{R}$. We abuse notation and denote by $\mathcal{R}(w)$ and $\mathcal{R}(j)$ the regions of worker $w$ and job $j$, respectively.

To capture the objective of assigning candidates with different degrees to jobs with matching levels, we formulate the following definition when evaluating a matching with respect to our simplified setting.

\begin{defn}
A matching $\mu$ \textbf{prioritizes more high-levels} than $\mu'$ if $|\mu|\geq |\mu'|$, $|\mu(J)\cap W^A|\geq |\mu'(J)\cap W^A|$, $|\mu(J^A)|\geq |\mu'(J^A)|$, and $|\mu(J^{AB})|\geq|\mu'(J^{AB})|$. A matching $\mu$ \textbf{prioritizes high-levels} if there does not exist another matching $\mu'$ that prioritizes more high-levels than $\mu$.
\end{defn}

In other words, a matching $\mu$ prioritizes more high-levels than another matching $\mu'$ if (i) $\mu$ does not match fewer workers than $\mu'$, (ii) $\mu$ does not match fewer workers from $W^A$ than $\mu'$, (iii) $\mu$ does not match fewer workers from $W^A$ to jobs in $J^A$ than $\mu'$, and (iv) it does not match fewer workers from $W^B$ to jobs in $J^{AB}$, except when compensated by an increase in workers from $W^A$ matched to these jobs.\footnote{Notice that the interpretations in items (iii) and (iv) are derived from more than one condition in the definition. More specifically, (iii) comes from $|\mu(J)\cap W^A|\geq |\mu'(J)\cap W^A|$ in combination with $|\mu(J^A)|\geq |\mu'(J^A)|$, which guarantees that, given that weakly more workers in $W^A$ should be matched to jobs under $\mu$, these should be matched to jobs in $J^A$, as opposed to $J^{AB}$. Conditional on that, (iv) draws on $|\mu(J^{AB})|\geq|\mu'(J^{AB})|$ to guarantee that workers in $W^B$ who could be matched to jobs in $J^{AB}$ or $J^B$ are matched to the former under $\mu$.} This definition captures, with relatively simple conditions, the objective of prioritizing the matching of higher-level jobs over lower-level ones, and of higher-level candidates over lower-level ones.\footnote{The definition is also implicitly treating jobs in $J^{AB}$ as ``mid-level'' jobs, which have lower priority than those in $J^{A}$, but higher than jobs in $J^{B}$.} With (i), it avoids considerations that involve not only which types of workers to which jobs, but also trade-offs between the overall number of jobs matched and their characteristics. When all eligibility constraints are satisfied, it guarantees the trade-offs that remain between which candidates and jobs to match are resolved toward matching more advanced jobs to more advanced candidates.   

Finally, for our historical analysis, it will be helpful to consider two different compatibility correspondences. The first one, denoted by $\mathcal{C}^{-}$, considers only the constraint on eligibility. That is, for any $w\in W$ and $j\in J$, $j\not\in\mathcal{C}^{-}(w)$ if and only if $w\in W^A$ and $j\in J^B$ or $w\in W^B$ and $j\in J^A$.
The second one, denoted by $\mathcal{C}^{+}$, considers both the constraints of eligibility and rule of avoidance. That is, for any $w\in W$ and $j\in J$, $j\in\mathcal{C}^{+}(w)$ if and only if $j\in\mathcal{C}^{-}(w)$ and $\mathcal{R}(j)\neq \mathcal{R}(w)$. By comparing the results we obtain with these two configurations we are able to evaluate the properties that the procedures have for a designer who considers all of the constraints and how they interact to one that was only focusing on eligibility.

\subsection{Models of the procedures}
\label{subsec:ProceduresModel}

We use the baseline model of markets, assignment plans, and assignment arrangements to represent the five historical procedures described in Section~\ref{sec:AssignmentsThroughHistory}. Given the nature of the historical sources, some simplifications are unavoidable; we make additional assumptions only when they are consequential for the analysis and keep them as parsimonious as possible. When a procedural detail is undocumented for a given period but explicitly described for earlier or later procedures, we assume continuity. Whenever feasible, we adopt assumptions that make procedures more comparable, so that our results do not hinge on inessential modeling differences. All modeled elements correspond directly to features documented in the historical sources.

\subsubsection{Song and Ming Procedures}
\label{subsubsec:SingleMarketProcedures}

The \Song, First \Ming, and Second \Ming procedures share a common assignment structure: all workers and jobs are processed in a single sequential assignment market, with no partitioning by job type. In all three cases, workers are considered in the order $\langle W^A, W^B\rangle$, where workers in each category are ranked by examination results. Once a worker is matched to a job, the assignment is final and not revised.

Formally, all three procedures share the same \emph{assignment arrangement},
\[
\varphi^{SM}=\left\{\langle W^A, W^B\rangle,\langle J\rangle\right\},
\]
which specifies that workers in $W^A$ are considered first, followed by workers in $W^B$, and that all jobs are contained in a single tube.

The procedures differ only in the specification of the \emph{assignment plan} $\langle \succ^W,(\succ^J_w)_{w\in W}\rangle$.

In all three procedures, the worker order $\succ^W$ satisfies: for any $w,w'\in W^I$ ($I\in\{A,B\}$), $w\succ^W w'$ if and only if $w$ obtained a higher examination result than $w'$.

The procedures differ in how the job orders $(\succ^J_w)_{w\in W}$ are generated:
\begin{itemize}
    \item \textbf{Song procedure.} For each worker $w$, $\succ^J_w$ represents the worker’s strict preference ordering over jobs. Moreover, we assume that a chosen job is approved whenever it is compatible with $w$ under eligibility and the rule of avoidance.
    
    \item \textbf{First Ming procedure.} Workers do not choose jobs. Instead, $\succ^J_w$ is interpreted as the priority order used by the Ministry of Personnel. For workers in $W^A$, jobs in $J^A$ are ranked above jobs in $J^{AB}$; for workers in $W^B$, jobs in $J^{AB}$ are ranked above jobs in $J^B$. Within each category, the specific ordering is immaterial for our results.
    
    \item \textbf{Second Ming procedure.} For each worker $w$, the job order $\succ^J_w$ is independently drawn from the uniform distribution over all permutations of $J$, representing a random draw from a tube containing all jobs compatible with $w$.
\end{itemize}

Thus, the three procedures are identical in their assignment arrangement and worker priority structure, and differ only in the source of the job ordering: preferences in the \Song procedure, discretionary assessment in the First \Ming procedure, and randomized draws in the Second \Ming procedure.

\subsubsection{The First \Qing Procedure}
\label{subsubsec:QingOneModel}

The First \Qing procedure departs from earlier ones by partitioning the assignment into independent markets by job type. Jobs are grouped into three categories $J^A$, $J^{AB}$, and $J^B$, and for each category a corresponding set of eligible workers is selected. Workers are assigned to these markets based on exam rank and, where applicable, quota rules.

Formally, the assignment arrangement consists of three sequences of tubes:
\[
\varphi^{Q1-A}=\{\langle W^A_1\rangle,\langle J^A\rangle\},\quad
\varphi^{Q1-AB}=\{\langle W^A_2\cup W^B_2\rangle,\langle J^{AB}\rangle\},\quad
\varphi^{Q1-B}=\{\langle W^B_1\rangle,\langle J^B\rangle\},
\]
where the worker sets are constructed so that the number of workers equals the number of jobs in each market.

Within each market, assignments are sequential and determined by drawing lots. Both the worker order $\succ^W$ and the job orders $(\succ^J_w)_{w\in W}$ are independently drawn uniformly at random. Matches violating the rule of avoidance are discarded and redrawn, and unmatched workers or jobs are carried over to the next appointment cycle.

\subsubsection{The Second \Qing Procedure}
\label{subsubsec:QingTwoModel}

The Second \Qing procedure modifies the First \Qing procedure only in the internal ordering of workers within each partitioned market. For each market, workers are split into two tubes: those who are incompatible with some jobs in the corresponding job tube due to the rule of avoidance, and those who are compatible with all jobs. The former are prioritized.

Formally, for each worker set $W^I_k$ used in the First \Qing procedure ($I\in\{A,B\}$, $k\in\{1,2\}$), define
\[
W^{I-P}_k=\{w\in W^I_k:\mathcal C^+(w)\cap J^I\neq J^I\},\qquad
W^{I-NP}_k=W^I_k\setminus W^{I-P}_k.
\]
The assignment arrangement is obtained from that of the First \Qing procedure by replacing each worker tube with an ordered pair $\langle W^{I-P}_k,W^{I-NP}_k\rangle$. The assignment plan remains otherwise unchanged: worker and job orders are independently drawn uniformly at random.

\section{Analysis of the assignment procedures and the changes}
\label{sec:AnalyzingTheProcedures}

This section presents the main comparative results of the paper. We evaluate how the matchings produced by the five historical procedures perform in terms of minimizing unfilled jobs and prioritizing high-level candidates, and how the successive changes in these procedures affect these outcomes.

Our analysis is conducted under two compatibility correspondences. Under $\mathcal{C}^{-}$, compatibility is determined solely by eligibility. Under $\mathcal{C}^{+}$, compatibility incorporates both eligibility and the rule of avoidance. Comparing outcomes under these two correspondences allows us to isolate the role played by the interaction between these constraints and to assess how the same procedural change can have qualitatively different effects depending on which constraints are operative.

\subsection{Main Results}
\label{subsec:AnalysisHistoricalMech-MainResults}

We begin with two theorems that summarize the paper’s central findings regarding the historical procedures. The first shows that reforms intended to improve assignment outcomes can backfire once the rule of avoidance is taken into account, even when evaluated on the same market and assignment plan. The second identifies the final \Qing reform as qualitatively distinct, yielding a monotonic improvement in match cardinality. Crucially, these results are comparative: they evaluate reforms relative to earlier procedures for the same market and assignment plan.

\begin{thm}
\label{prop:FromSongToMingOneLessAssortative}
When considering compatibility in terms of both eligibility and the rule of avoidance $\mathcal{C}^{+}$, there exist markets and assignment plans for which the matching produced by both the \Song and the Second \Ming procedures prioritizes high-levels and leaves no worker unmatched, whereas under the First \Ming, the First \Qing, and the Second \Qing procedures the matching fails to prioritize high-levels and leaves half of the workers unmatched.
\end{thm}
\begin{proof}
Let there be three regions: $X$, $Y$, and $Z$. Let, moreover, $W^A_X$ $\left(J^A_X\right)$ be the set of workers (jobs) from region $X$, and the same for $Y$ and $Z$.  Consider then a market where $|W^A|= 2n$, and $|J^{A}|= |J^{AB}|=n$, with $n \in \mathbb{N}^+$. Among the $A$ workers, $n$ of them come from region $X$, and the other $n$ of them come from region $Y$. Among the jobs, $n$ of the $AB$ jobs are from region $X$, and $n$ of the $A$ jobs are from region $Z$.   

Consider the following assignment plan $\left\langle \succ^W,\left(\succ_w^J\right)_{w\in W}\right\rangle$, where for every $w_y\in W^A_Y$ and $w_x\in W^A_X$ $w_y\succ^W w_x$; for every $w\in W$, $j_{x}\in J^{AB}_X$, $j_{z}\in J^{A}_Z$, $j_x \succ_w^J j_{z}$. That is, all workers in $W_Y^A$ are assigned before all workers in $W_X^A$; and for every worker, all jobs in $J_X^{AB}$ are assigned before all jobs in $J_Z^A$. 

The matching $\mu$ produced by both the \Song and Second \Ming procedures matches the $n$ workers in $W^A_Y$ to the $n$ jobs in $J^{AB}_X$, and the $n$ workers in $W^A_X$ to the $n$ jobs in $J^{A}_Z$. Notice that $\mu$ prioritizes high-levels and minimizes unfilled jobs.

The matching produced by the First and Second \Qing depends on the specific way in which the workers in $W^A$ are partitioned.\footnote{In practice, this followed the candidates' exam grades, so that those with higher grades were put into the sequence of tubes with the $J^{A}$ jobs, and those with lower grades to the remaining jobs. The arrangement we use relies, therefore, on the workers from region $X$ having lower grades.} We will consider the case in which the sequence of tubes for the $A$ jobs is $\left\{\left\langle W^{A}_Y \right\rangle,\left\langle J^{A}_Z\right\rangle\right\}$ and for the $AB$ jobs $\left\{\left\langle W^{A}_X \right\rangle,\left\langle J^{AB}_X\right\rangle\right\}$.

The matching $\mu'$ produced by the First \Ming and the First and the Second \Qing procedures---with the sequence of tubes above---matches the $n$ workers in $W^A_Y$ to the $n$ jobs in $J^{A}_Z$, but leaves the $n$ workers in $W^A_X$ and the $n$ jobs in $J^{AB}_X$ unmatched. 
Notice that the matching $\mu$ leaves no worker or job unmatched, and prioritizes more high-levels than $\mu'$.
\end{proof}

Theorem \ref{prop:FromSongToMingOneLessAssortative} formalizes a central tension in the historical evolution of the procedures. Although the First \Ming and First \Qing reforms were motivated by concerns about improving the prioritization of high-level candidates, their interaction with the rule of avoidance can lead to strictly worse outcomes along both dimensions. Moreover, the magnitude of this deterioration can be substantial.

A key implication of this result is that the direction of a reform’s effect depends critically on which constraints are taken into account. For the same market and assignment plan, a procedural change can appear beneficial when evaluated under eligibility alone, yet be harmful once avoidance constraints are introduced. Importantly, this reversal is not driven by different realizations of randomness or different priority orders, but by the interaction between constraints.

Our second main result shows that the final change in the historical sequence does not suffer from this drawback.

\begin{thm}
\label{prop:FromQingOneToQingTwo}
Consider any market and assignment plan, and let $\mu^{Q1}$ and $\mu^{Q2}$ denote the matchings produced by the First and Second \Qing procedures, respectively. Then $|\mu^{Q2}| \geq |\mu^{Q1}|$. Moreover, there exist markets and assignment plans for which $|\mu^{Q2}| > |\mu^{Q1}|$.
\end{thm}
\begin{proof}
    
Let $\varphi^1=\left\{\left\langle W \right\rangle,\left\langle J \right\rangle\right\}$ be any of the sequences of tubes in the First \Qing, and then let  $\varphi^2=\left\{\left\langle W^{P},W^{NP} \right\rangle,\left\langle J \right\rangle\right\}$ be the corresponding sequences of tubes in the Second \Qing.\footnote{For the purposes of this proof, the categories of the workers and jobs involved in these sequences of tubes is inconsequential, since from the perspective of these categories every worker is compatible with every job.}
    
First, notice that since workers in the set $W^{NP}$ are compatible with every job in $J$, the only way that some worker in $W^{NP}$ is left unmatched in both procedures is if all jobs are matched. Therefore, one way to show that $|\mu^{Q2}|\geq|\mu^{Q1}|$ is to show that every worker in $W^{P}$ who is matched to a job under the First \Qing is also matched to a job under the Second \Qing. 

Let $\left\langle \succ^W,\left(\succ_w^J\right)_{w\in W}\right\rangle$ be any assignment plan, and $\succ^W_{P}$ be the ranking $\succ^W$ restricted to the set of prioritized workers $W^{P}$. In addition, let $n_P=|W^P|$, and $rank_i(\succ)$ be the $i$-th highest-ranking worker in the ranking $\succ$. 
We now show by induction that, for every worker $w\in W^{P}$, the set %$J^*$ 
of jobs that are still available in the step where $w$ is matched to a job under the First \Qing procedure, denoted by $J^{*Q1}_w$, is a subset of the jobs that are available when matching $w$ under the Second \Qing, denoted by $J^{*Q2}_w$. 

\textbf{Base}: Let $w^1=rank_1(\succ^W_{P})$. Under the Second \Qing Procedure, the set of jobs still available when matching $w^1$, the highest-ranking worker among the prioritized workers, is $J$. Therefore, $J^{*Q1}_{w^1}\subseteq J^{*Q2}_{w^1}=J$. 

\textbf{Step}: Assume that for every $1\leq i\leq k<n_P$ and $w^i=rank_i(\succ^W_{P})$, $J^{*Q1}_{w^i}\subseteq J^{*Q2}_{w^i}$, and consider a worker  $w^{k+1}=rank_{k+1}(\succ^W_{P})$. Suppose, for a contradiction, that there is a job $j^*$ that is available when matching him under the First \Qing procedure but not when matching him under the Second \Qing procedure, that is, formally, $j^*\in J^{*Q1}_{w^{k+1}}$ but $j^*\not\in J^{*Q2}_{w^{k+1}}$. 
By induction assumption and the fact that $j^*\in J^{*Q1}_{w^{k+1}}$, $j^*\in J^{*Q1}_{w^{k}}$  and therefore $j^*\in J^{*Q2}_{w^{k}}$. By construction of the Second \Qing procedure, the only way that this can happen is if the job $j^*$ is matched to worker $w^{k}$ in the Second \Qing procedure. That is, $j^*=top^{\succ^J_{w^{k}}}(J^{*Q2}_{w^{k}})$. 
Moreover, we know as a result of the immediate assignment by the procedure that 
$J^{*Q1}_{w^{k+1}}\subseteq J^{*Q1}_{w^k}$. 
By the contradiction assumption $j^*\in J^{*Q1}_{w^{k+1}}$, we therefore have $j^*\in J^{*Q1}_{w^{k}}$. 
On the other hand, the contradiction assumption $j^*\in J^{*Q1}_{w^{k+1}}$ implies that $j^*\neq top^{\succ^J_{w^{k}}}(J^{*Q1}_{w^{k}})$. But this is a contradiction with $j^*=top^{\succ^J_{w^{k}}}(J^{*Q2}_{w^{k}})$, since by induction assumption $J^{*Q1}_{w^k}\subseteq J^{*Q2}_{w^k}$.

Finally, to observe that there are some situations where $|\mu^{Q2}|>|\mu^{Q1}|$, we consider the following market. There are three regions: $X$, $Y$, and $Z$. Moreover, let $W^A_X$ ($J^A_X$) be the set of workers (jobs) from region $X$, and the same for $Y$ and $Z$. Let $|W^A_X|=|W^A_Y|=n$ and $|J^A_X|=|J^A_Z|=n$.

Consider the following assignment plan $\left\langle \succ^W,\left(\succ_w^J\right)_{w\in W}\right\rangle$, where for every $w_y\in W^A_Y$ and $w_x\in W^A_X$ $w_y\succ^W w_x$, for every $w\in W$, $j_{x}\in J^{A}_X$, $j_{z}\in J^{A}_Z$, $j_z \succ_w^J j_{x}$. 

The matching $\mu$ produced by the First \Qing procedure matches the $n$ workers in $W^A_Y$ to the $n$ jobs in $J^{A}_Z$, but leaves the remaining $n$ workers in $W^A_X$ and $n$ jobs in $J^A_X$ unmatched. 

The matching $\mu$ produced by the Second \Qing procedure matches the $n$ workers in $W^A_X$ to the $n$ jobs in $J^{A}_Z$, and the $n$ workers in $W^A_X$ to the $n$ jobs in $J^A_X$.
\end{proof}

Theorem \ref{prop:FromQingOneToQingTwo} isolates the Second \Qing reform as the only change in the sequence that yields a weakly monotonic improvement in match cardinality for all markets and assignment plans. The material that follows decomposes these results by separately analyzing the roles of eligibility and avoidance constraints, and by explaining the mechanisms through which the reversals in Theorem \ref{prop:FromSongToMingOneLessAssortative} arise.

\subsection{Mechanisms and comparative statics}
\label{subsec:PropertiesProceduresAndComparativeStatics}

To understand why the reversals identified in Theorem \ref{prop:FromSongToMingOneLessAssortative} can arise, we first illustrate a generic limitation of sequential assignment mechanisms. This limitation is independent of the historical details and applies even in very simple markets. We then formalize these forces through a sequence of propositions that separately analyze eligibility constraints and their interaction with the rule of avoidance.

All the procedures we evaluate were implemented by matching workers to jobs sequentially. Despite its transparency and simplicity, this type of mechanism may suffer from some shortcomings due to the ``greedy'' way in which they process the matchings.

More specifically, both the minimization of unfilled jobs and the prioritization of high-levels are properties that, in general, are incompatible with matching each worker with an arbitrary compatible job, one at a time. To see this, consider the example below.\footnote{While the example considers eligibility constraints, one can easily construct an analogous one using only the rule of avoidance.}

\begin{example}\label{ex:inefficiencyc-}
Consider two workers $\{w_a,w_b\}$ and two jobs $\{j_a,j_{ab}\}$, where $w_a$ is eligible for both jobs but $w_b$ is eligible only for $j_{ab}$. If $w_a$ is matched first to $j_{ab}$, then $w_b$ remains unmatched; yet a feasible matching assigns $\mu(w_a)=j_a$ and $\mu(w_b)=j_{ab}$.
\end{example}

The example illustrates a generic limitation of sequential (``greedy'') assignment: processing one worker at a time can waste scarce compatible jobs and fail to minimize unfilled positions. A natural remedy would be to allow ex post revisions (e.g., exchanges), but such revisions were apparently resisted in practice, plausibly because they undermine transparency, especially when assignments are determined by lots. Indeed, an exchange rule was proposed during the early years of the Second \Ming procedure\footnote{In 1602, in a correspondence to the court by Li Dai---the minister of the Ministry of Personnel at the time, it was suggested that a candidate who either draws an incompatible job or ends up with no compatible jobs left would be able to exchange his assignment with some candidate matched with a compatible job in a mutually acceptable way \citep{Wu1609}.}, but there is no indication it was implemented; and mid-\Qing sources mention adding extra candidates to reduce unfilled jobs \citep{zeli1886}. These concerns are consistent with the fact that producing maximum matchings under such constraints is a nontrivial combinatorial problem (e.g., the Hungarian method dates to the 1950s \citep{kuhn1955hungarian}, and online matching remains an active literature \citep{feng2014online}).

Our first result shows that, when considering only eligibility constraints, the changes in the procedures were consistent with the objective of minimizing unfilled jobs.

\begin{prop}
\label{prop:CardinalEfficiencyAllProceduresCMinus}
When considering only compatibility in terms of eligibility $\mathcal{C}^{-}$, the First \Ming, First \Qing, and Second \Qing procedures produce matchings that minimize unfilled jobs, while there are markets and assignment plans for which the \Song and the Second \Ming procedures produce matchings that do not minimize unfilled jobs.
\end{prop}
\begin{proof}
First, recall that the First \Ming procedure tries to match workers to jobs that fit their degrees, following the order to match as many workers in $W^A$ first to jobs in $J^A$ and then to $J^{AB}$, and afterwards, match as many workers in $W^B$ to jobs in the remaining $J^{AB}$ jobs and then to $J^{B}$. Therefore, (i) jobs in $J^A$ will only be left unmatched if $|W^A|<|J^A|$, (ii) jobs in $J^{AB}$ will only be left unmatched if every worker is matched, and (iii) jobs in $J^{B}$ will only be left unmatched if all workers in $W^B$ are matched. One can clearly see that when some job is left unmatched, the three observations above imply that there is no alternative matching that could match that job without leaving another unmatched. 

Second, notice that every sequence of tubes under the First and Second \Qing procedures contains workers and jobs that are mutually compatible with respect to $\mathcal{C}^{-}$. Since the number of workers and jobs are equal in each one of these sequences of tubes, every worker is matched to a job, and therefore the matching produced by these two procedures minimizes unfilled jobs. 

Next, to see that the matchings produced by the \Song and Second \Ming procedures might not minimize unfilled jobs, consider a market where there is an equal number of $A$ workers, $B$ workers, $A$ jobs, and $AB$ jobs, and $|W^A|=|W^B|=|J^{A}|=|J^{AB}|=n$ with $n \in \mathbb{N}^+$. 
%$|W^A|=|W^B|=|J^{A}|=|J^{AB}|=n>0$. 
Consider moreover an assignment plan, where $A$ workers are assigned before $B$ workers, and for every worker, $AB$ jobs are considered before $A$ jobs. That is,  
% $\left\langle \succ^W,\left(\succ_w^J\right)_{w\in W}\right\rangle$, where 
for every $w_a\in W^A$ and $w_b\in W^B$, $w_a\succ^W w_b$; and for every $w\in W$, $j_{ab}\in J^{AB}$ and $j_{a}\in J^{A}$, $j_{ab}\succ_w^J j_{a}$. 

The matchings produced by both the \Song and the Second \Ming procedures match the $n$ workers in $W^A$ to the $n$ jobs in $J^{AB}$, leaving all the workers in $W^B$ and jobs in $J^A$ unmatched. Notice, however, that there is an alternative matching that matches every worker and every job: match all workers in $W^A$ to the jobs in $J^A$, and all workers in $W^B$ to the jobs in $J^{AB}$.
\end{proof}

The proposition below shows that when rule of avoidance is taken into consideration, there are configurations in which none of the procedures used minimizes unfilled jobs.

\begin{prop}
\label{prop:CardinalEfficiencyAllProceduresCPlus}
When considering compatibility in terms of both eligibility and rule of avoidance $\mathcal{C}^{+}$, there are markets and assignment plans for which all the procedures produce matchings that do not minimize unfilled jobs.
\end{prop}

\begin{proof}
Let there be three regions: $X$, $Y$, and $Z$. Moreover, let $W^A_X$ ($J^A_X$) be the set of workers (jobs) from region $X$, and define the same for $Y$ and $Z$ respectively. Consider a market where there is an equal number of $A$ workers and $A$ jobs, with $|W^A|=|J^{A}|=2n$ and $n \in \mathbb{N}^+$, 
and there are no $B$ workers nor $B$ or $AB$ jobs. 
In addition, within $A$ workers, there are $n$ from region $X$ and $n$ from region $Y$. 
On the job side, within $A$ jobs, we have $n$ from region $Z$, $n-1$ from region $X$, and 1 from region $Y$. 
Moreover, consider an assignment plan 
where for every $w_y\in W^A_Y$ and $w_x\in W^A_X$,  $w_y\succ^W w_x$; and for every $w\in W$, $j_{z}\in J^{A}_Z$, $j_{x}\in J^{A}_X$, and $j_{y}\in J^{A}_Y$, $j_z \succ_w^J j_{x}$ and $j_x \succ_w^J j_{y}$. Put differently, the assignment plan is such that all $A$ workers from region $Y$ are assigned before all $A$ workers from region $X$, and for every worker, jobs from region $Z$ are considered before jobs from region $X$, which are considered before jobs from region $Y$. 

Then, the matchings produced by all the procedures that we are evaluating match the $n$ workers in $W^A_Y$ to the $n$ jobs in $J^A_Z$, and one worker in $W^A_X$ to the job in $J^A_Y$, leaving all the remaining $n-1$ workers in $W^A_X$, and the $n-1$ jobs in $J^A_X$ unmatched. Notice, however, that there is an alternative matching that matches every worker and every job: match $n-1$ workers in $W^A_X$ to jobs in $J^A_Z$, one worker in $W^A_X$ to the job in $J^A_Y$, $n-1$ workers in $W^A_Y$ to jobs in $J^A_X$, and one worker in $W^A_Y$ to a job in $J^A_Z$. Therefore, the matchings produced by all the procedures do not minimize unfilled jobs when taking both eligibility and rule of avoidance into account. 
\end{proof}

Next, we consider the extent to which the matchings produced by all the procedures prioritize high-levels. As before, we first discuss the case where we only take eligibility into account and find that changes were also consistent with this objective.

\begin{prop}
\label{prop:AssortativenessAllProceduresCMinus}
When considering only compatibility in terms of eligibility $\mathcal{C}^{-}$, the First Ming\textsuperscript{$\mathbb{2}$}, First Qing\textsuperscript{$\mathbb{3}$}, and Second \Qing procedures produce matchings that prioritize high-levels, while there are markets and assignment plans for which the \Song and the Second \Ming procedures produce matchings that do not prioritize high-levels.
\end{prop}

\begin{proof}

First, recall in the First and Second \Qing procedures, workers and jobs are partitioned into three sequences of tubes in which they are all compatible with respect to $\mathcal{C}^{-}$. Since the number of workers and jobs is equal in each one of these sequences of tubes, every worker is matched to a job. Moreover, every job in $J^A$ is matched to a worker in $W^A$. These two facts together imply that this matching prioritizes high-levels.

Next, consider the First \Ming procedure. Regardless of the assignment plan, the following observations can be made about the matching produced:

\begin{enumerate}
    \item An unmatched job in $J^A$ is possible only if there are less $A$ workers than $A$ jobs, $|W^A|<|J^A|$. 
    \item If $|J^A|<|W^A|\leq |J^{A}|+|J^{AB}|$, then every job in $J^A$ is matched to a worker in $W^A$, and every worker in $W^A$ is matched either to a job in $J^A$ or $J^{AB}$. 
    \item If $|W^A|>|J^{A}|+|J^{AB}|$, then every job in $J^A$ and $J^{AB}$ is matched to a worker in $W^A$.
    \item If a job in $J^{AB}$ is left unmatched, then every worker is already matched to a job. If a worker in $W^B$ is left unmatched, then every job in $J^{AB}\cup J^{B}$ is already matched to a worker.
\end{enumerate}
These observations combined imply that the matching produced by the First \Ming procedure prioritizes high-levels.

Finally, to see that the \Song and the Second \Ming procedures may not prioritize high-levels, consider the case in which $|W^A|=|J^B|=|J^{AB}|=n$ and $n \in \mathbb{N}^+$. Moreover, let the assignment plan be such that for every $w_a\in W^A$ and $w_b\in W^b$, $w_a\succ^W w_b$; and for every $w\in W$, $j_{ab}\in J^{AB}$, $j_{a}\in J^{A}$, $j_{ab} \succ_w^J j_{a}$. That is, all workers in $W^A$ are matched before all workers in $W^B$, and for every worker, all jobs in $J^{AB}$ are considered before all jobs in $J^A$.
Then, the matching produced by the \Song and the Second \Ming procedures matches the $n$ workers in $W^A$ to the $n$ jobs in $J^{AB}$, leaving all the remaining $n$ workers in $W^B$, and the $n$ jobs in $J^A$ unmatched. Notice, however, that an alternative matching that matches every worker in $W^A$ to the jobs in $J^A$ and every worker in $W^B$ to the jobs in $J^{AB}$ prioritizes more high-levels than the one produced either by the \Song or the Second \Ming procedure.
\end{proof}

As for the minimization of unfilled jobs, the rule of avoidance also affects the prioritization of high-levels. The following shows that when we consider both eligibility and rule of avoidance as compatibility constraints, all procedures can result in matchings that do not prioritize high-levels. This result follows immediately from markets constructed from the proof of Proposition \ref{prop:CardinalEfficiencyAllProceduresCPlus}, and therefore the proof is omitted.

\begin{prop}
\label{prop:AssortativenessAllProceduresCPlus}
When considering compatibility in terms of both eligibility and rule of avoidance $\mathcal{C}^{+}$, 
there are markets and assignment plans for which all the procedures produce matchings that do not prioritize high-levels.
\end{prop}
 
Notice that if we consider only the results involving compatibility in terms of eligibility, we would conclude that both the changes from the \Song procedure to the First \Ming procedure and from the Second \Ming procedure to the First \Qing procedure would unambiguously improve the matchings that were produced, both in terms of the minimization of unfilled jobs and prioritization of high-levels. However, when we compare with respect to compatibility in terms of eligibility and rule of avoidance, the same conclusion can no longer be made. 

Figure \ref{fig:summaryOfChanges} summarizes the above results. Note that we know from Proposition \ref{prop:CardinalEfficiencyAllProceduresCPlus}, despite the improvement yielded by the Second \Qing procedure, in general the Second \Qing procedure may still not match as many candidates as possible. In the section below, however, we show that in certain markets the addition of a single tube of jobs could result in a procedure that produces matchings that minimizes unfilled jobs and prioritize high-levels in \textit{every realization of chance}.

\begin{figure}
    \centering
    \includegraphics[scale=0.75]{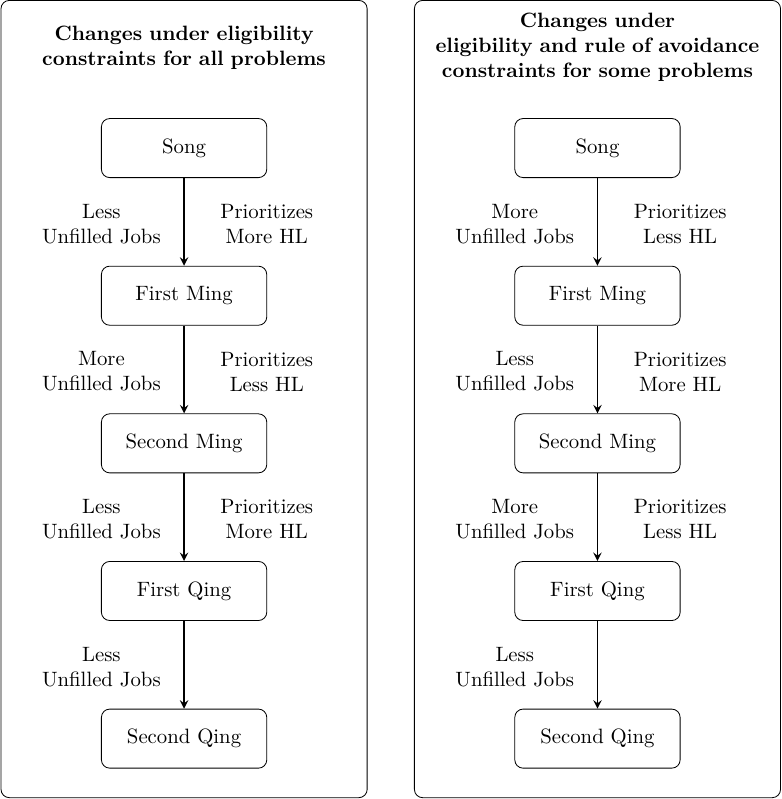}
    \caption{Summary of changes in terms of minimization of unfilled jobs and prioritization of high-levels (HL)}
    \label{fig:summaryOfChanges}
\end{figure}

%%%%%%%%%%%%%%%%%%%%%%%%%%%%%%%%%%%%%%%%%%%%%%%%%%%%%%%%%%%%%%%%%%%%%%
%%%%%%%%%%%%%%%%%%%%%%%%%%%%%%%%%%%%%%%%%%%%%%%%%%%%%%%%%%%%%%%%%%%%%%
%%%%%%%%%%%%%%%%%%%%%%%%%%%%%%%%%%%%%%%%%%%%%%%%%%%%%%%%%%%%%%%%%%%%%%
%%%%%%%%%%%%%%%%%%%%%%%%%%%%%%%%%%%%%%%%%%%%%%%%%%%%%%%%%%%%%%%%%%%%%%

\section{Optimal lots-drawing mechanisms}
\label{sec:NewMechanisms}

Having analyzed the historical evolution of assignment procedures, we now turn from evaluation to design. This section addresses the problem of proactively constructing a procedure that achieves the objectives identified in our analysis: minimizing unfilled jobs and prioritizing high-levels, while adhering to eligibility and rule of avoidance constraints.

The central challenge lies in the sequential, or ``greedy,'' nature of the lots-drawing process. The objective of minimizing unfilled jobs requires the procedure to produce a matching of maximum cardinality for every realization of chance. This is a complex constraint, as a seemingly innocuous match made early in the sequence can preclude the formation of a maximum matching later on. The following example illustrates this inherent complexity.\footnote{In this section we restrict our attention to assignment arrangements consisting of a single sequence of tubes. That is, we are not going to consider solutions that are based on segmenting the markets into different sub-markets which are then matched using separate sequences of tubes.}

\begin{example}
\label{example:motivatingUrnsDesign}
Consider a market with three workers, $w_1, w_2, w_3$, and three jobs, $j_1, j_2, j_3$, where the subscript indicates the region of the worker or job. The rule of avoidance is the only active constraint. In this market, there are two feasible matchings of maximum size:

\begin{align*}
    \mu_1 &= \{ (w_1, j_2), (w_2, j_3), (w_3, j_1) \} ,\\
    \mu_2 &= \{ (w_1, j_3), (w_2, j_1), (w_3, j_2) \}.
\end{align*}

Achieving a maximum matching imposes a dependency between assignments; a locally permissible match can prevent a globally optimal outcome.

For instance, suppose the procedure begins by matching worker $w_1$. The compatible jobs for $w_1$ are $j_2$ and $j_3$. Let's assume $w_1$ is matched to $j_2$. If worker $w_2$ is considered next, the available compatible jobs are $j_1$ and $j_3$. Suppose $w_2$ is then matched to $j_1$. This sequence of matches leaves worker $w_3$ and job $j_3$. As they are from the same region, the rule of avoidance prohibits this match, leaving both unassigned. The resulting matching has a size of two, whereas a matching of size three was possible.

This issue is not unique to a specific ordering of workers. The problem is symmetric. If, for instance, the procedure started by matching $w_3$ to $j_1$, a subsequent match of $w_1$ to $j_3$ would similarly leave $w_2$ and $j_2$ stranded and unmatchable. The greedy nature of a sequential procedure thus makes guaranteeing a maximum matching a non-trivial design challenge.
\end{example}

Guaranteeing a maximum matching with a sequential procedure thus requires embedding a form of conditionality into the process, a feature absent in the simple lots-drawing mechanisms discussed so far.\footnote{The lots-drawing procedure used in the UEFA Champions League group stage draw is an instance of a mechanism that implements this conditionality. After a team is drawn from a bowl, a computer calculates the set of groups into which that team can be placed while still permitting a valid final assignment for all remaining teams. The subsequent draw is then made only from that valid set of groups, ensuring that a maximum matching is always reached \citep{boczon2018goals}.} The designer's objective, therefore, is to craft a sequence of tubes that restricts and prioritizes matches in such a way that the greedy outcome is always maximal. This is far from a trivial enterprise. The structure of a maximum matching often depends on information about the overall structure of the compatibility graph, and not on local neighborhoods. In fact, when we consider arbitrary compatibility correspondences, the number of tubes that are necessary to guarantee the minimization of unfilled jobs using lots-drawing procedures can be arbitrarily large, as shown by the proposition below. 

\begin{prop}\label{prop:unboundedNumberOfTubes}
For every integer $k\ge 2$ there exists a market $\langle W,J,\mathcal C\rangle$ such that for any sequence of tubes
\[
\varphi=\Big\{\langle W_1,\ldots,W_n\rangle,\ \langle J_1,\ldots,J_m\rangle\Big\}
\]
with $m<k$, there exists an assignment plan $\langle \succ^W,(\succ^J_w)_{w\in W}\rangle$ for which the matching produced by $\varphi$ does not minimize unfilled jobs.
\end{prop}

\begin{proof}
Fix $k\ge 2$ and let

\begin{eqnarray*}
    W & = & \{u_1,\dots,u_{k-1}\}\ \cup\ \{w_1,\dots,w_{k-1}\}\ \cup\ \{z_k\},\\
J & = & \{a_1,\dots,a_k\}\ \cup\ \{x_1,\dots,x_{k-1}\}.
\end{eqnarray*}

Define the compatibility correspondence by
\[
\mathcal C(u_i)=\{a_i,a_{i+1}\}\ (1\le i\le k-1),\qquad
\mathcal C(z_k)=\{a_k\},\qquad
\mathcal C(w_i)=\{x_i\}\ (1\le i\le k-1),
\]
and no other pairs are compatible. The feasible matching
\[
\mu^\star(u_i)=a_i\ (1\le i\le k-1),\quad \mu^\star(z_k)=a_k,\quad \mu^\star(w_i)=x_i\ (1\le i\le k-1)
\]
fills every job, so any matching that minimizes unfilled jobs must have size $|J|=2k-1$.

\smallskip
Let $\{J_1,\ldots,J_m\}$ be any partition of $J$ into job tubes with $m<k$; the worker tubes $\{W_1,\ldots,W_n\}$ being arbitrary and fixed. For each $i\in\{1,\ldots,k\}$ let $p(i)\in\{1,\ldots,m\}$ be the unique index with $a_i\in J_{p(i)}$. Since $m<k$, there exists $i\in\{1,\ldots,k-1\}$ with $p(i)\ge p(i+1)$. Let $i$ be the smallest such index, then
\[
p(1)<p(2)<\cdots<p(i)\quad\text{and}\quad p(i)\ge p(i+1).
\]

\smallskip
We will pick \emph{only} (i) within-tube worker orders $\succ^W$ and (ii) job orders $(\succ^J_w)_{w\in W}$ so that, no matter how the workers are split across the tubes $W_1,\ldots,W_n$, when each $u_j$ is eventually processed the following happens:
\begin{enumerate}
\item For every $j<i$, $u_j$ takes $a_j$ (because $a_j$ lies in an earlier job tube than $a_{j+1}$).
\item $u_i$ takes $a_{i+1}$ (because $a_{i+1}$ is in a job tube no later than $a_i$, and we break ties toward $a_{i+1}$).
\item Consequently, the only potential takers of $a_i$---namely $u_{i-1}$ and $u_i$---are already committed to $a_{i-1}$ and $a_{i+1}$ when they are processed, so $a_i$ is never matched.
\end{enumerate}
Notice that this argument is independent of which worker tube contains each $u_j$: we only use the rule that a worker, when processed, must choose from the earliest job tube that still contains a compatible job, and the fact that assignments are irrevocable.

\smallskip
\textbf{Assignment plan.} For the within-tube worker order $\succ^W$, order those present so that among the $u$'s the relative order is
\[
u_1\succ^W u_2\succ^W\cdots\succ^W u_{i-1}\succ^W u_i\succ^W z_k\succ^W u_{i+1}\succ^W\cdots\succ^W u_{k-1},
\]

For the job orders, set for each $u_j$:
\[
\succ^J_{u_j}:\quad
\begin{cases}
a_j \succ^J_{u_j} a_{j+1}, & j<i\ ,\\[2pt]
a_{i+1} \succ^J_{u_i} a_i, & j=i\ ,\\[2pt]
a_{j+1} \succ^J_{u_j} a_{j}, & j>i\ .
\end{cases}
\]

Orders for $z_k$ and $w_\ell$ are immaterial.

\smallskip
\textbf{Execution}
Consider the procedure induced by $\varphi$ and the assignment plan above.

\emph{Claim 1.} For every $j<i$, when $u_j$ is processed (in whichever worker tube he resides), the earliest job tube containing a compatible job is $J_{p(j)}$, because $p(j)<p(j+1)$ and earlier $a$'s have not been compatible with $u_j$. Thus $u_j$ takes $a_j$, and $a_j$ is removed.

\emph{Claim 2.} When $u_i$ is processed, his earliest compatible job lies in a tube no later than $J_{p(i+1)}$ (because $p(i)\ge p(i+1)$). By $\succ^J_{u_i}$ he takes $a_{i+1}$, not $a_i$.

\emph{Claim 3.} The only workers compatible with $a_i$ are $u_{i-1}$ and $u_i$. By Claim 1, $u_{i-1}$ takes $a_{i-1}$ when processed; by Claim 2, $u_i$ takes $a_{i+1}$. Therefore $a_i$ is never assigned.

Finally, $z_k$ (whenever processed) takes $a_k$, and each $w_\ell$ takes $x_\ell$; these choices never interact with the $a$-chain.

\smallskip
\textbf{Conclusion.} The produced matching $\mu$ leaves $a_i$ unmatched, hence $|\mu|\le |J|-1=2k-2<2k-1=|\mu^\star|$. Therefore $\mu$ does not minimize unfilled jobs. Because the worker tubes were arbitrary, the failure holds for \emph{any} worker-tube partition whenever $m<k$ job tubes. This proves the proposition.
\end{proof}

Let us now revisit the market in Example \ref{example:motivatingUrnsDesign} and consider the following sequence of tubes:
\[
\varphi=\left\{\left\langle \{w_1\},\{w_2,w_3\}\right\rangle,\left\langle \{j_2,j_3\},\{j_1\}\right\rangle\right\}.
\]

This structure first considers worker $w_1$ for a match, then the workers in $\{w_2,w_3\}$. More importantly, it forces any worker being considered to first attempt a match with an available job in the first job tube, $\{j_2,j_3\}$. This architecture directly prevents the failure mode identified earlier. For instance, if $w_1$ is matched to $j_2$, worker $w_2$ must next consider the remaining job in the first tube, $j_3$, before he is allowed to consider $j_1$. Since $w_2$ is compatible with $j_3$, they are matched, leaving the compatible pair $(w_3, j_1)$ to form the final match. The arrangement makes it impossible to form the problematic ``cross-matches'' (like $(w_1, j_2)$ and $(w_2, j_1)$) that result in non-maximum outcomes.

Surprisingly, for problems involving only the rule of avoidance, this two-tube structure can be generalized to markets of any size, under a simple condition.

Let $\{r_1,\ldots,r_k\}$ be the set of regions to which workers and jobs belong, and let $W_i$ and $J_i$ denote the set of workers and jobs from region $r_i$, respectively. We use $W_{-i}$ to denote the set of all workers not in $W_i$ (i.e., $W \setminus W_i$) and, similarly, $J_{-i}$ for the set of all jobs not in $J_i$. Without loss of generality, let $W_1$ be a region with the highest number of workers (i.e., $|W_1|\geq |W_i|$ for all $i > 1$).

\begin{defn}
A market is \textbf{regular} if $|W_{-1}|\geq |J_{-1}|$.
\end{defn}

This regularity condition requires that the total number of workers from all regions except the largest one ($W_1$) is at least as large as the total number of jobs from those same regions. Now consider the following sequence of tubes:
\begin{align*}
\varphi^\ast=\left\{\left\langle W_1,W_{-1} \right\rangle,\left\langle J_{-1},J_1 \right\rangle\right\}.
\end{align*}

This sequence prioritizes workers from the region with the most workers ($W_1$) to be matched first. Critically, it also prioritizes jobs from all other regions ($J_{-1}$) to be considered first for any match.\footnote{If a region contains workers but no jobs, its workers are included in the second worker tube, $W_{-1}$. If a region contains jobs but no workers, its jobs are included in the first job tube, $J_{-1}$.}

\begin{thm}
\label{thm:TwotubesOnEachSide}
If the market is regular, the matchings produced by the sequence of tubes $\varphi^\ast$ minimize unfilled jobs for any assignment plan when considering only the rule of avoidance.
\end{thm}

\begin{proof}

In a first step, we show that all jobs in the first tube, $J_{-1}$, will be matched. Suppose that some job in  $J_{-1}$ is left unmatched. Let this unmatched job be located in region $r_t$, a region that, by the definition of $J_{-1}$, does not have the largest number of workers. That is, $j \in J_{t}$, with $t\geq 2$. It must then be that all workers who can be matched to $j$, including all workers in $W_1$ and all workers in $W_{-1}$ except $W_t$ (who are incompatible with jobs in $J_t$), are already matched to jobs in $J_{-1}$.  
This implies that, for any $t\geq 2$, the following condition must hold:
\[
\left | J_{-1}\right | > \left | W_1\right | + \left | W_{-1,t}\right |, 
\]
where $W_{-1,t}= \cup_{k \not\in \{1,t\}} W_k$. The condition says that there are more jobs in $J_{-1}$ than there are total workers excluding those from region $t$, where the unmatched job $j$ is located. Since $W_1\geq W_t$, then 
we must have: 

\[
\left | J_{-1}\right | > \left | W_t\right | + \left | W_{-1,t}\right | = \left |  W_{-1} \right|,
\]
which is a contradiction with the market being regular.

In a second step, we consider the jobs in the second tube, $J_{1}$. 
Notice that if no job in $J_{1}$ is left unmatched, then all jobs are matched to workers, therefore minimizing unfilled jobs. Suppose then that some job in $J_{1}$ is left unmatched. A first observation is that it must be the case that all workers in $W_{-1}$ (who are compatible with $J_1$) are matched. Otherwise, there would be no job left unmatched in $J_{1}$. 

There are then two cases to consider. First, there is no worker in $W_{1}$ left unmatched. If that's the case, then all workers are matched and therefore the matching minimizes unfilled jobs. 

Second, there is some worker in $W_{1}$ left unmatched. Since he is in the first tube of workers, it must be that workers in $W_{1}$ exhausted all jobs in $J_{-1}$. In that case, however, by the time the workers in $W_{-1}$ are drawn, all jobs available (those in $J_{1}$) are compatible with them. 

Here, if $\left|J_{1}\right| \leq \left|W_{-1}\right|$, we have a contradiction with some job in $J_{1}$ being left unmatched. In this case, all jobs are matched. 

If $\left|J_{1}\right|> \left| W_{-1}\right|$, then it is still possible that some job in $J_{1}$ is left unmatched. To see that the matching is still minimizing unfilled jobs in this case, we will represent the market as a bipartite graph, where $W$ are vertices on one side and each vertex represents a worker, $J$ are the vertices on the other side and each vertex represents a job, and an edge connects $w\in W$ and $j\in J$ if and only if they are compatible, i.e., they are from different regions. A matching is then a set of the edges such that each worker or job can only appear in at most one of the edges in that set. An \textbf{augmenting path} is a path with an odd number of edges in which both ends are unmatched vertices and the edges alternate between edges inside and outside the matching. 
The following lemma is useful to show our result. 

\begin{lem}[Berge's Lemma]\label{lem:berge}
A matching minimizes unfilled jobs if and only if it contains no augmenting path. 
\end{lem} 

Suppose the resulting matching does not minimize unfilled jobs. Then, by Berge's Lemma, this implies there is an augmenting path connecting an unmatched worker in $W_{1}$ and an unmatched job in $J_{1}$. An alternating path that starts at a worker in $W_{1}$, however, never includes an element of $J_{1}$. To observe this, consider an unmatched worker in $W_{1}$. That worker is connected to the jobs in $J_{-1}$. Since the next edge in the augmenting path must be in the matching, it connects next to a worker in $W_{1}$ (since all jobs in $J_{-1}$ are matched to workers in $W_{1}$). Therefore, if we keep adding edges to the augmenting path we will always alternate between vertices in $W_{1}$ and $J_{-1}$ . Therefore, there is no augmenting path connecting an unmatched worker in $W_{1}$ and an unmatched job in $J_{1}$---a contradiction with the matching not minimizing unfilled jobs.
\end{proof}

Interestingly, this extension involving the use of two ordered tubes of jobs is not only a natural extension of the original procedures, but similar to an extension that was briefly experimented while implementing the Second \Ming procedure. In it, instead of having a single tube of jobs, there were four tubes, each representing four divisions of the empire---northwest, northeast, southwest, and southeast \citep{will2002creation}. Candidates would draw first from the tube containing jobs from the division where their home province was before proceeding to the other ones. The intention was to offer a compromise between the use of the rule of avoidance (which was still in place) and an attempt to place candidates in locations that were not too far from their hometown. The four-tube system, however, stopped briefly after its introduction. One reason seemed to be that corruption became rife after the drawing procedure became more complicated \citep{shen1619wanli}. The need for considering proximity was no longer mentioned in the \Qing discussions. The procedure we propose, however, is not only an interesting extension of the Second \Qing procedure, but also a simple and transparent proposal for producing, in the present, maximum matchings under these type of constraints.

The two-by-two structure of $\varphi^\ast$ is, in fact, minimal. As the following proposition shows, sequences with only one tube on either side cannot guarantee a maximum matching, even in the well-behaved class of regular markets.

\begin{prop}
\label{prop:twotwominimal_corrected_final_v4}
In a market with at least three regions, and considering rule of avoidance, no sequence of tubes with only one worker tube or only one job tube can guarantee the minimization of unfilled jobs for all possible assignment plans, even if the market is regular.
\end{prop}

\begin{proof}
Fix the market $\langle W,J,\mathcal{C}\rangle$, where $W=\{w_1,w_2,w_3\}$, $J=\{j_1,j_2,j_3\}$, where worker $w_i$ and job $j_i$ belong to region $r_i$. Note that this market is regular. We will consider both cases, showing that for any sequence with one tube of workers or one tube of jobs, there is a market and an assignment plan that does not minimize unfilled jobs.

\medskip\noindent\textbf{(i) One worker tube.}
Let $\varphi=\{\langle W\rangle,\langle J_1,\ldots,J_m\rangle\}$ and let
\[
t^\star \in \{1,\ldots,m\}\quad\text{be the largest index with }J_{t^\star}\cap J\neq\emptyset
\]
(the last nonempty job tube in this market). Choose any $j_\ell\in J_{t^\star}$ and let $w_\ell$ be the unique worker in region $r_{\ell}$.

\emph{Assignment plan.}
Because there is a single worker tube, we may choose $\succ^W$ arbitrarily; set $w_\ell$ last in $\succ^W$.
For job orders, define for every $w\in W$ the strict order $\succ_w^J$ so that all $w$ rank each of their compatible jobs in $J_{t^\star}\setminus\{j_\ell\}$ above $j_\ell$ (so $j_\ell$ is last among $w$'s compatible jobs inside the last nonempty tube).

\emph{Execution.}
Process the first two workers in $\succ^W$ (both are $\neq w_\ell$). Each such worker $w$ identifies $a^*$ as the first tube index with $\mathcal{C}(w)\cap J_{a^*}\neq\emptyset$. If $a^*<t^\star$, he takes a job in $\bigcup_{a<t^\star}J_a$ by (a). If $a^*=t^\star$, he takes a job in $J_{t^\star}\setminus\{j_\ell\}$ by (b). In either case, neither of the first two workers takes $j_\ell$.
Hence, when the last worker $w_\ell$ is considered, the unique remaining job is $j_\ell$, which is incompatible with $w_\ell$. Therefore $\mu(w_\ell)=\emptyset$ and exactly two matches are formed, i.e.\ $|\mu|=2<3$. Since a perfect matching exists in this market, $\mu$ does not minimize unfilled jobs.

\medskip\noindent\textbf{(ii) One job tube.}
Let $\varphi=\{\langle W_1,\ldots,W_n\rangle,\langle J\rangle\}$ and let
\[
u^\star \in \{1,\ldots,n\}\quad\text{be the largest index with }W_{u^\star}\cap W\neq\emptyset
\]
(the last nonempty worker tube in this market). Choose any $w_\ell\in W_{u^\star}$ and let $j_\ell$ be the unique job in region $r_\ell$.

\emph{Assignment plan.}
Because there is a single job tube, every worker $w$ chooses from all of $J$, subject to $\succ_w^J$ and availability.  
Choose $\succ^W$ so that $w_\ell$ is last among $W_{u^\star}$.  
For job orders, define $\succ_w^J$ for the first two workers (those processed before $w_\ell$) so that each of them strictly prefers some compatible job in $J\setminus\{j_\ell\}$ to $j_\ell$. This ensures that, whenever they are considered, they will be matched to jobs in $J\setminus\{j_\ell\}$ rather than $j_\ell$. The ranking of $j_\ell$ relative to other jobs is irrelevant for them once those choices are made. The order for $w_\ell$ is immaterial since he will be left with only $j_\ell$.

\emph{Execution.}
Process the workers in the order given by $\succ^W$. The first two workers each select a job in $J\setminus\{j_\ell\}$, since these jobs are ranked above $j_\ell$ in their orders and are compatible. When the last worker $w_\ell$ is processed, the only job remaining is $j_\ell$, which is incompatible with him. Hence $\mu(w_\ell)=\emptyset$ and exactly two matches are formed, i.e.\ $|\mu|=2<3$. Since a perfect matching exists in this market, $\mu$ does not minimize unfilled jobs.

\medskip
In both cases we constructed an assignment plan \emph{as a function of the last nonempty tube on the opposite side} (jobs in (i), workers in (ii)), in accordance with the execution rule and ignoring empty tubes. Since a perfect matching exists in the market, the produced feasible matching is not of maximum cardinality.
\end{proof}

Proposition \ref{prop:twotwominimal_corrected_final_v4} demonstrates that the $\varphi^\ast$ arrangement is optimal in its simplicity; it uses the minimum number of tubes required to achieve its goal in regular markets. This result, combined with Theorem \ref{thm:TwotubesOnEachSide}, provides a valuable insight. The institutional framework of tubes and sequential drawing, developed centuries ago for the Chinese imperial civil service, contains the core components of a surprisingly simple and elegant solution to the complex problem of generating random maximum matchings. With a small but critical structural modification, the historical procedure can be transformed into a mechanism that is transparent, easily implementable, and provably optimal under a broad class of conditions.

So far, we have only considered rule of avoidance constraints. Next, we would like to also consider eligibility constraints and its associated objective of prioritizing high-levels. Before we proceed, however, it is important to note that for any market and any matching that minimizes unfilled jobs, there is a sequence of tubes that reproduces that matching.

\begin{prop}
\label{prop:trivialDeterministicMechanism}
For every market $\langle W,J,\mathcal{C}\rangle$ and every matching $\mu\in\mathcal{M}$ that minimizes unfilled jobs, there exists a sequence of tubes $\varphi$ such that, for any assignment plan $\langle \succ^W,(\succ_w^J)_{w\in W}\rangle$, the outcome of $\varphi$ under that plan is exactly $\mu$.
\end{prop}

\begin{proof}
Let $\mu$ be a feasible matching of maximum cardinality for the market $\langle W,J,\mathcal{C}\rangle$.  
Define the matched sets of agents
\[
W^+ := \{w\in W : \mu(w)\neq \emptyset\}, \quad 
J^+ := \{j\in J : \mu(j)\neq \emptyset\},
\]
and the unmatched sets $W^-:=W\setminus W^+$, $J^-:=J\setminus J^+$.  
Because $\mu$ is a matching, $|W^+|=|J^+|$ and $\mu$ induces a bijection between $W^+$ and $J^+$.  

Let $W^+=\{w_1,\dots,w_{|W^+|}\}$ be any enumeration of the matched workers, and denote $j_i:=\mu(w_i)$ for each $i$.  
We now construct the following sequence of tubes:
\[
\varphi^{\mu} \;=\; \Big\{\, 
  \langle \{w_1\},\{w_2\},\ldots,\{w_{|W^+|}\}, W^- \rangle ,\;
  \langle \{j_1\},\{j_2\},\ldots,\{j_{|J^+|}\}, J^- \rangle 
  \,\Big\}.
\]

\emph{Execution.}  
Now simulate the procedure defined by $\varphi^{\mu}$ for any assignment plan $(\succ^W,(\succ_w^J)_{w\in W})$.  
For $i=1,\dots,|W^+|$, when worker $w_i$ is processed, the first tube of jobs still containing a compatible job is $\{j_i\}$. Since $j_i\in \mathcal{C}(w_i)$, $w_i$ is matched to $j_i$ and $j_i$ is removed.  
After all $w_1,\dots,w_{|W^+|}$ are processed, the only remaining workers are those in $W^-$ and the only remaining jobs are those in $J^-$.  
By maximality of $\mu$, for every $w\in W^-$ and every $j\in J^-$ we have $j\notin \mathcal{C}(w)$. Thus all such workers are left unmatched, and so are the jobs in $J^-$.  

Therefore, the outcome of $(\varphi^{\mu},\succ^W,(\succ_w^J)_{w\in W})$ coincides exactly with $\mu$. Since $\mu$ was an arbitrary maximum-cardinality matching, the claim follows.
\end{proof}

Proposition \ref{prop:trivialDeterministicMechanism} shows that the problem of producing matchings that minimize unfilled jobs always has a trivial solution when using lots-drawing mechanisms: one can always design a ``degenerate'' sequence of tubes that deterministically implements a particular maximum matching. This, however, is clearly not what a desirable solution looks like, as such a procedure removes all randomness and symmetry across agents.

We will now return to the historical setup of regions and categories, and introduce a basic ``desirability'' property for sequences of tubes.  

\begin{defn}[Anonymous Sequence of Tubes]
A sequence of tubes 
\[
\varphi=\big\{\langle W_1,\ldots,W_k\rangle,\;\langle J_1,\ldots,J_{\ell}\rangle\big\}
\] 
is said to be \textbf{anonymous} if the following conditions hold:
\begin{enumerate}
    \item For all regions $r$ and categories $c\in\{A,B\}$, all workers in $W^c\cap W_r$ appear in the same worker tube $W_t$ for some $t\in\{1,\ldots,k\}$.
    \item For all regions $r$ and categories $c\in\{A,AB,B\}$, all jobs in $J^c\cap J_r$ appear in the same job tube $J_s$ for some $s\in\{1,\ldots,\ell\}$.
\end{enumerate}
That is, workers (resp.\ jobs) that are indistinguishable in terms of their region and category are not separated across different tubes.
\end{defn}

Anonymity is a natural property: it prevents workers or jobs that are identical from the perspective of the mathematical formulation of the market from being treated differently by the mechanism.\footnote{Note that the sequence of tubes in Theorem \ref{thm:TwotubesOnEachSide} is anonymous.}

Next, we show that the three objectives of the designer---(i) minimizing unfilled jobs, (ii) prioritizing high-levels, and (iii) anonymity---are mutually incompatible when both rule of avoidance and eligibility constraints are present.

\begin{prop}
\label{prop:incompatibilityMaximizingAndPrioritizing}
Suppose there are at least three regions, and consider both eligibility and rule of avoidance constraints. Then no anonymous sequence of tubes simultaneously minimizes unfilled jobs and prioritizes high-levels, even when markets are regular.
\end{prop}

\begin{proof}
Let $\varphi=\{\langle W_1,\ldots,W_k\rangle,\langle J_1,\ldots,J_\ell\rangle\}$ be an anonymous sequence of tubes. We consider two exhaustive cases depending on the relative ordering of tubes of $A$ jobs and $AB$ jobs.

\medskip\noindent\textbf{Case 1. Every $A$-job tube precedes every $AB$-job tube.}

Consider a regular market $\langle W,J\rangle$ with two workers 
\[
W=W^A=\{w_a,w_b\}, \quad a\neq b,
\]
where $w_a$ and $w_b$ are both $A$-workers from distinct regions $r_a,r_b$.  
Let the jobs be 
\[
J=\{j_b,j_c\}, \quad J^A=\{j_c\}, \quad J^{AB}=\{j_b\}.
\]
By anonymity, $j_c$ is placed in some $A$-job tube $J_\alpha$ and $j_b$ in some $AB$-job tube $J_\beta$, with $J_\alpha$ ordered before $J_\beta$.  

Suppose $\succ^W$ is such that $w_a \succ^W w_b$. When $w_a$ is processed, the first nonempty tube with a compatible job is $J_\alpha$, so $w_a$ is matched to $j_c$. Then $w_b$ is processed, and the only remaining job is $j_b$, which is incompatible with $w_b$ (same region). Thus $\mu(w_b)=\emptyset$ and the outcome has $|\mu|=1$.  
However, the matching $\mu'=\{(w_a,j_b),(w_b,j_c)\}$ is feasible and has size 2. Hence $\mu$ does not minimize unfilled jobs.  

Therefore, no anonymous sequence in which all $A$-job tubes precede all $AB$-job tubes can minimize unfilled jobs.

\medskip\noindent\textbf{Case 2. Some $AB$-job tube precedes an $A$-job tube.}

Let $J_\alpha$ precede $J_\beta$ in $\varphi$, with $j_a\in J_\alpha$ an $AB$ job in region $r_a$ and $j_b\in J_\beta$ an $A$ job in region $r_b$.  

Consider the regular market $\langle W,J\rangle$ with three regions $r_a,r_b,r_c$ and
\[
W=\{w_c,w_a\},\qquad J=\{j_a,j_b\},
\]

where $w_c$ is an $A$-worker from region $r_c$, and $w_a$ is a $B$-worker from region $r_a$. Worker $w_c$, therefore, is compatible with both $j_a$ ($AB$ job in region $r_a$) and $j_b$ ($A$ job in region $r_b$). Worker $w_a$, however, is incompatible with $j_a$ (same region $r_a$) and ineligible for $j_b$.  Hence the only worker with feasible options is $w_c$, and any feasible matching has cardinality at most one.  

\emph{Execution.}  
Because $J_\alpha$ (an $AB$ job tube) precedes $J_\beta$ (an $A$ job tube), $w_c$ will first encounter $j_a$ and be matched to it. The resulting matching $\mu=\{(w_c,j_a)\}$ has size one and therefore maximizes cardinality.  

However, there exists another feasible matching $\mu'=\{(w_c,j_b)\}$, also of size one, that assigns the $A$-worker $w_c$ to the $A$ job $j_b$, thereby prioritizing more high-levels. Thus, although $\mu$ maximizes cardinality, it fails to prioritize high-levels.  

We conclude that no anonymous sequence in which some $AB$-job tube precedes an $A$-job tube can satisfy both minimization of unfilled jobs and prioritization of high-levels.

\medskip
In both cases, the objectives of minimizing unfilled jobs and prioritizing high-levels are incompatible under anonymity. Since the cases exhaust all possible anonymous sequences of tubes, the claim follows.
\end{proof}

Proposition \ref{prop:incompatibilityMaximizingAndPrioritizing} therefore shows that while rule of avoidance alone admits mechanisms that are anonymous and efficient in terms of minimizing unfilled jobs, the simultaneous imposition of eligibility and high-level prioritization destroys this compatibility: no anonymous sequence of tubes can satisfy all three desiderata at once.

It is important to note, however, that Proposition \ref{prop:incompatibilityMaximizingAndPrioritizing} does not rule out the existence of meaningful solutions under alternative notions of desirability, or under stronger structural assumptions on the market. In particular, one might refine the definition of desirable procedures in ways that admit sequences which are anonymous, minimize unfilled jobs, and prioritize high-levels in restricted environments.  

One of the main challenges lies in the dimensionality of the design space. Even when considering only rule of avoidance constraints, the number of distinct anonymous sequences of tubes increases very rapidly with the number of regions: for example, there are $9$ such sequences with $2$ regions, $169$ with $3$ regions, and $5625$ with $4$ regions. This combinatorial explosion illustrates that the mechanism designer faces a highly complex search problem.  

At the same time, Proposition \ref{prop:trivialDeterministicMechanism} guarantees that, with respect to the objective of minimizing unfilled jobs, feasible solutions always exist: one can always construct a ``degenerate'' sequence of tubes that deterministically implements a maximum matching. The more substantive challenge---and the intellectually rewarding task for the designer---is to identify and characterize procedures that achieve the intended objectives while also satisfying normative desiderata such as anonymity and prioritization of high-levels, under assumptions that are both natural and practically relevant.

\section{Discussion}\label{sec:conclusion}

This paper studied the evolution of civil service assignment procedures used over more than a millennium to allocate qualified candidates to government posts in imperial China. By reconstructing these procedures from historical sources and embedding them in a unified matching framework, we were able to compare their outcomes and clarify the trade-offs underlying successive reforms. A central finding is that changes motivated by plausible administrative objectives did not necessarily improve performance once multiple constraints were taken into account. In particular, the sequential and irrevocable nature of the procedures could generate unintended consequences when eligibility and avoidance constraints interacted.

Beyond the specific historical setting, a key contribution of the paper is methodological. Formal modeling provides a way to infer the implicit objectives and concerns of institutional designers even when these are not fully articulated in the historical record. By tracing the logical implications of documented rules, the analysis links observed procedural details to the properties of the resulting allocations, helping to interpret why certain reforms were introduced, retained, or abandoned. At the same time, the framework disciplines empirical and historical analysis by making explicit which constraints and implementation choices shape assignment outcomes, a point that is particularly relevant when historical randomness is used for identification.

The long-run perspective also sheds light on institutional persistence and change. Examining which design features survived repeated reforms, and which were modified or discarded, provides information about the structural properties that remained compatible with evolving administrative and political constraints. In this sense, the historical record serves as a form of revealed design experimentation, offering guidance for evaluating alternative assignment rules beyond their immediate performance.

Finally, the analysis points to broader implications for the design of matching mechanisms that rely on \emph{public randomization}, where outcomes are determined through observable and verifiable draws of lots. Such mechanisms occupy an intermediate position between fully deterministic rules and unrestricted stochastic allocation. Our results highlight both the potential and the limits of sequential random procedures, suggesting directions for future work on constrained matching in richer environments. Similar design issues arise in contemporary settings—such as the assignment of judges to cases, students to schools, or public resources across regions—where transparency and verifiability make public randomization an attractive institutional tool.

\section*{Declaration of generative AI and AI-assisted technologies in the writing process}\label{sec:AIDeclaration}
During the preparation of this work the author(s) used ChatGPT in order to explore ideas, make minor edits and rephrase some sentences. After using this tool/service, the author(s) reviewed and edited the content as needed and take(s) full responsibility for the content of the publication.

\bibliographystyle{ecca}
\bibliography{civil}
\end{document}